\documentclass[12pt]{article}

\usepackage{enumerate}
\usepackage{amsmath}
\usepackage{amssymb}
\usepackage{graphicx}
\usepackage{xcolor}
\usepackage{booktabs}
\usepackage{adjustbox}
\usepackage{threeparttable}
\usepackage{threeparttablex}
\usepackage{array}
\usepackage{etoolbox}
\usepackage{natbib} 
\setcitestyle{authoryear,maxnames=2,minnames=1}
\usepackage[linesnumbered,ruled,vlined]{algorithm2e}
\usepackage{amsthm}
\usepackage{url}

\usepackage[utf8]{inputenc}
\usepackage{textcomp}
\usepackage{dsfont}
\usepackage{array}
\usepackage{newtxtext}
\usepackage{newtxmath}
\usepackage{anyfontsize}
\usepackage{makecell}
\usepackage[ruled,vlined,linesnumbered]{algorithm2e}
\usepackage[colorlinks=true,linkcolor=blue,citecolor=blue,urlcolor=blue]{hyperref} 
\usepackage{authblk} 
\newtheorem{assumption}{Assumption}
\newtheorem{lemma}{Lemma}
\newtheorem{definition}{Definition}
\newtheorem{theorem}{Theorem}
\newtheorem{example}{Example}

\newtheorem{remark}{Remark}

\begin{document}

\def\spacingset#1{\renewcommand{\baselinestretch}%
{#1}\small\normalsize} \spacingset{1}


\title{\bf A Conformal Selection Framework for Individual Treatment Beneficiaries with Auxiliary External Data}

\author[1]{Jiajun Liu}
\author[1,2]{Ke Zhu}
\author[1]{Xiaofei Wang\footnote{Address for correspondence: Xiaofei Wang, Department of Biostatistics and Bioinformatics, Duke University, Durham, NC 27710, U.S.A. Email: xiaofei.wang@duke.edu}}

\affil[1]{\small Department of Biostatistics and Bioinformatics, Duke University, Durham, NC 27710, U.S.A.}

\affil[2]{\small Department of Statistics, North Carolina State University, Raleigh, NC 27695, U.S.A.}

\date{}
\maketitle

\begin{abstract}
  Identifying patients who are likely to benefit from a treatment is an important goal in precision medicine and can help guide follow-up trials, enrichment designs, and individualized treatment decisions. Although randomized controlled trials (RCTs) provide reliable evidence about treatment efficacy, they are usually powered to estimate average treatment effects rather than patient-level benefit. Meanwhile, artificial intelligence and machine learning methods offer flexible tools for estimating heterogeneous treatment effects, especially when augmented by external data sources such as real-world data (RWD). However, in practice, these estimated effects are often translated into decisions through simple ranking or thresholding rules, which can ignore uncertainty in treatment effect estimation and the multiplicity that arises when many candidate patients are evaluated simultaneously. Motivated by this, we propose a model-agnostic conformal inference framework for uncertainty-aware beneficiary selection. The framework reformulates CATE-based treatment-benefit selection as a multiple-testing problem. For each candidate, we test whether the conditional treatment benefit exceeds a clinically meaningful threshold and construct a conformal $p$-value using RCT-based calibration. These $p$-values are then adjusted by the Benjamini-Hochberg procedure to control the false discovery rate (FDR) among selected beneficiaries. To improve efficiency when RCT sample sizes are limited, external data (e.g. RWD) can be used to train flexible treatment effect models, while conformal calibration remains anchored in the RCT data. Because the framework is model-agnostic, it can be paired with a broad class of learners, including conventional machine learning algorithms and emerging tabular foundation models. Simulations show that the framework maintains FDR control, with power depending on the base learner and external-data comparability. A case study in early-stage non-small-cell lung cancer illustrates how the method can identify candidate profiles with evidence of benefit from limited resection to reduce overtreatment.

\end{abstract}

\bigskip
\noindent%
\vfill

\newpage
\spacingset{1.9} 

\section{Introduction}\label{sec:introduction}
Precision medicine aims to identify which patients are likely to benefit from a treatment and to use this information to guide individualized treatment decisions, follow-up trials, enrichment designs, and other targeted study designs \citep{FDA2025PrecisionMedicine}. Randomized controlled trials (RCTs) remain the gold standard for evaluating treatment effectiveness in a target population \citep{Deaton2018}. However, RCTs are usually designed and powered to estimate average treatment effects rather than patient-level treatment benefit \citep{Kontopantelis2018}. In many clinical settings, the average effect does not tell the full story. A treatment may be beneficial on average, while some patients receive little benefit or may even be harmed. Conversely, even when the overall average treatment effect (ATE) is not statistically significant, clinically meaningful benefit may still exist for certain patients \citep{Baum2017LookAHEAD,Inoue2024MedicaidHeterogeneous}. Therefore, it is important to consider approaches that move beyond overall treatment efficacy and toward identifying patients with evident individual-level benefit.

This goal has both scientific and ethical importance. If individualized or heterogeneous treatment effects are ignored, potential beneficiaries may be prevented from accessing effective therapies, whereas non-beneficiaries may be exposed to unnecessary risks, burdens, and costs \citep{Dahabreh2016UsingGroupData,Kent2018,Kontopantelis2018,kent_predictive_2020,angus_heterogeneity_2021,Inoue2024MLHTE}. With RCTs serving as an important source of reference evidence \citep{Kent2018}, precision medicine aims to support individualized treatment decisions by identifying which patients are more likely to benefit from a given therapy. In this paper, we focus on the effect modeling paradigm for treatment benefit assessment. In contrast to risk modeling, which asks ``Who is high-risk?'' based on baseline prognosis, effect modeling asks ``Who benefits more from treatment?'' by directly targeting treatment effect heterogeneity \citep{Kent2018}.

A concrete example arises in early-stage non-small-cell lung cancer (NSCLC), where recent randomized trials have shown that limited resection can be non-inferior to lobectomy for selected patients with small, peripheral, node-negative tumors \citep{saji2022segmentectomy, Altorki2023, Altorki2024}. This population-level conclusion, however, does not imply that the two procedures are interchangeable for every patient. Some patients may benefit more from the less invasive procedure, whereas others may still benefit from standard lobectomy. This clinical setting illustrates the need to move from average treatment comparisons to uncertainty-aware, patient-level treatment-benefit selection. We return to this example in the case study.

An extensive literature has developed for estimating heterogeneous treatment effects, most commonly through the conditional average treatment effect (CATE), \(\tau(x)=E\{Y(1)-Y(0)\mid X=x\}\). Early work primarily modeled treatment-covariate interactions through outcome regression, often referred to as regression-based effect modeling \citep{GailSimon1985,Robins1986,YusufWittesProbstfieldTyroler1991}. More recent approaches use flexible statistical and machine learning methods, including penalized regression for high-dimensional covariates \citep{Imai2013,Du2021}, causal forests \citep{Wager2018,Seibold2018,Doubleday2022}, Bayesian causal forests \citep{Shen2016,Hahn2020}, meta-learners such as the S-learner \citep{Hill2020,Montoya2021,Conzuelo2021}, T-learner \citep{Athey2016}, X-learner and doubly robust learner \citep{Kunzel2019,Kennedy2023}, and R-learner \citep{Nie2021}, as well as tree-based and nonparametric methods such as causal trees, causal forests, and Bayesian additive regression trees \citep{Chipman2010BART,Athey2016,Wager2018}; see \citet{Inoue2024MLHTE} for a comprehensive overview.
Although these methods have greatly improved our ability to estimate treatment effect heterogeneity, accurate CATE prediction alone does not fully solve the beneficiary selection problem. In practice, predicted CATEs are often translated into treatment decisions by ranking patients or applying a fixed threshold. This strategy is simple, but it treats predicted effects as known without error and ignores the multiplicity that arises when many candidate patients are evaluated simultaneously \citep{kent_predictive_2020,Inoue2024MLHTE}. As a result, patients may be selected as likely beneficiaries based only on point estimates, and the proportion of falsely selected non-beneficiaries may be inflated. Therefore, CATE-based prediction should be connected to an inferential rule that assesses whether there is sufficient evidence of individual-level treatment benefit while controlling false selections. 

Motivated by this gap, we propose a conformal selection framework that reframes CATE-based beneficiary identification as an uncertainty-aware individual-level testing problem. For each candidate patient, we test whether the individual treatment benefit exceeds a clinically meaningful threshold. In many applications, this threshold may be zero, corresponding to any positive treatment benefit; in other settings, it may represent a minimum clinically important benefit. Rather than selecting patients solely by ranking or thresholding predicted treatment effects, we construct conformal $p$-values to measure the evidence for individualized treatment benefit. This testing formulation allows beneficiary identification to be based on patient-level evidence rather than point estimates alone. 
This use of conformal inference is well aligned with the needs of individual-level beneficiary selection. Conformal inference provides a way to assess patient-level evidence while maintaining finite-sample validity under exchangeability and avoiding parametric distributional assumptions \citep{shafer2008tutorial,lei2021conformal}. Because it is model-agnostic, the proposed framework can be adapted to a broad class of base learners, including conventional machine learning methods and emerging tabular foundation models such as TabPFN and TabICL \citep{zhang2025tabpfnmodelruleall,hollmann2025accurate,qu2025tabicl}. This flexibility is especially useful when treatment effect patterns are nonlinear, high-dimensional, or difficult to specify in advance. At the same time, instead of relying on large-sample approximations to the sampling distribution, conformal inference evaluates candidate patients against the empirical distribution of calibration scores. Through conformal $p$-values, this comparison provides an interpretable measure of evidence for individualized treatment benefit. Therefore, in contrast to ad hoc thresholding of CATE estimates, the proposed conformal-based testing formulation provides an explicit measure of uncertainty, allowing clinicians to ask not only ``who appears likely to benefit,'' but also ``how strong the evidence is to declare benefit.''

Our approach builds on recent work in conformal inference for causal and counterfactual quantities. \citet{lei2021conformal} developed conformal methods for counterfactual outcomes and individualized treatment effects (ITEs), and \citet{alaa2024conformal} proposed conformal meta-learners for constructing model-agnostic prediction intervals for ITEs. Related work has studied sensitivity analysis for ITEs under unmeasured confounding \citep{jin2023sensitivity}. Separately, \citet{Jin2023SelectionbyConP} introduced conformal $p$-values for selection problems with false discovery rate (FDR) control, providing a general foundation for error-controlled screening. More broadly, conformal inference has been increasingly used in treatment effect and counterfactual problems \citep{Cai2024,Chen2024Conformal,Jonkers2025,gao2025rolesurrogatesconformalinference,wang2025conformalinferenceindividualtreatment,liu2025hybridControl}, as summarized in a recent systematic review \citep{memmesheimer2025}. 
However, existing work does not directly address the beneficiary selection problem considered here. Conformal selection methods have primarily focused on observed outcomes, whereas precision medicine decisions often depend on treatment contrasts and individualized treatment benefit. At the same time, conformal methods for ITEs have largely focused on prediction intervals rather than on selecting beneficiaries with error-rate control. This distinction matters because the clinical question is usually not whether an outcome will be large or small, but whether a patient has sufficient evidence of treatment benefit relative to a clinically meaningful threshold. Therefore, a conformal testing procedure tailored to individual-level beneficiary selection is needed.

Two additional challenges arise in this setting. First, because beneficiary identification typically involves screening a group of candidate patients rather than testing a single individual, false selections can accumulate even when each candidate-level test is valid marginally \citep{Ranganathan2016}. We therefore use the Benjamini-Hochberg (BH) procedure to control the FDR \citep{benjamini1995,benjamini2001}. 
Second, RCTs are often too small to train reliable patient-level treatment effect models and are frequently criticized as underpowered for heterogeneous treatment effect analysis \citep{Kontopantelis2018,kent_predictive_2020}. This is particularly important in precision medicine, where treatment effect heterogeneity can be difficult to learn from small RCTs alone \citep{Inoue2024MLHTE}. To improve model training, we leverage external data sources, such as real-world data (RWD), which may expand covariate support and stabilize flexible effect models \citep{Kent2018,Lipkovich2024}. Because external data may not be exchangeable with the target trial population, we use external data only for model training and keep conformal calibration anchored in the RCT data.

In general, we propose a framework that (i) quantifies uncertainty in beneficiary identification using conformal $p$-values, (ii) explicitly addresses multiplicity through BH-based FDR control, and (iii) leverages external data to improve treatment effect model training under limited RCT sample sizes, while using RCT-based conformal calibration to preserve validity for the target trial population. Specifically, we reformulate CATE-based beneficiary identification as an individual-level hypothesis testing problem.
The proposed framework is intended for settings where investigators need to make selection decisions for a fixed group of candidate patients. For example, after an initial RCT, it may be used to identify patients for a follow-up target trial, an enrichment strategy, or a second-stage design. More broadly, it can support individualized treatment decisions by helping determine, before treatment assignment, which patients have evidence of benefit.
In this way, the framework links flexible treatment effect modeling with FDR-controlled patient selection in precision medicine.

The remainder of the paper is organized as follows. Section~\ref{sec:method} introduces the basic setup, pseudo ITE derivation, conformal $p$-value construction and FDR-controlled selection procedure. Section~\ref{sec:sim} presents simulation studies evaluating validity, power, and robustness under external-data bias. Section~\ref{sec:case_study} demonstrates how the proposed framework can be applied in practice through a case study. Section~\ref{sec:discussion} concludes with implications, limitations, and future directions.

\section{Method}\label{sec:method}
In this section, we propose a conformal selection framework for personalized treatment decision-making. The framework uses external data to strengthen treatment effect modeling, while maintaining explicit error control for beneficiary selection. Rather than directly ranking or thresholding estimated CATEs, we formulate beneficiary identification as an individual-level testing problem.

\subsection{Basic Setup}\label{sec:method_basicsetup}

Based on the potential outcomes framework \citep{Rubin1974}, for each individual, let $A \in \{0,1\}$ denote a binary treatment indicator, $Y \in \mathbb{R}$ the observed outcome, and $X \in \mathbb{R}^p$ a $p-$dimension vector of baseline covariates. For each individual, let \(Y(1)\) and \(Y(0)\) denote the potential outcomes under treatment and control, respectively.

We consider two data sources. The RCT dataset contains $n$ subjects with observed data $\mathcal{D}_{\mathrm{RCT}} = \{(Y_i, A_i, X_i)\}_{i=1}^n$, while the RWD dataset contains $m$ subjects with observed data $\mathcal{D}_{\mathrm{RWD}} = \{(Y_{n+i}, A_{n+i}, X_{n+i})\}_{i=1}^{m}$. Let \(S\in\{0,1\}\) denote the data source indicator, where \(S=1\) corresponds to the RCT and \(S=0\) to the RWD. In addition, we consider a candidate population $\mathcal{D}_{\mathrm{Test}} = \{X_{n+m+j}\}_{j=1}^{n_t}$, where treatment decisions have not yet been made and outcomes are not observed.  For these candidates, only baseline covariates are available. Throughout, we assume that the covariate distributions across data sources have sufficient overlap. 

Under the Stable Unit Treatment Value Assumption (SUTVA) \citep{Rubin1990}, the observed outcome satisfies
\begin{align*}
    Y_i = A_i Y_i(1) + (1-A_i) Y_i(0).
\end{align*}
Within the RCT, the treatment effect is identified through randomization mechanism. Assumption \ref{as:ident_rct} formalizes this condition and provides the basis for valid inference in our framework.

\begin{assumption}[Identification within the RCT]\label{as:ident_rct}
For individuals with $S=1$:
\begin{enumerate}
    \item \textit{Consistency}: $Y = A Y(1) + (1-A) Y(0)$;
    \item \textit{Positivity}: $0 < e(x) < 1$ for all $x$ such that $f_{X \mid S}(x \mid 1) > 0$;
    \item \textit{Randomization}: $Y(a) \perp A \mid (X, S=1)$ for $a \in \{0,1\}$.
\end{enumerate}
\end{assumption}

For the observational data (RWD), we impose the standard strong ignorability assumption to support their use in model training. Assumption \ref{as:strong_igno} rules out unmeasured confounding within RWD after conditioning on the observed covariates. Although this assumption is not required for identification or inference in the RCT, it helps justify the use of RWD for nuisance estimation and treatment effect modeling. If this assumption is violated, the external-data-assisted CATE model may be biased, which can affect efficiency, power, and the stability of beneficiary selection.

\begin{assumption}[Strong ignorability in RWD]\label{as:strong_igno}
For individuals with \(S=0\),
\[
(Y(1), Y(0)) \perp A \mid X.
\]
\end{assumption}

The individual treatment effect (ITE) is defined as $\tau_i = Y_i(1) - Y_i(0)$, which quantifies the benefit of treatment relative to control for individual $i$. Since only one potential outcome is observed for each individual, \(\tau_i\) is not directly identifiable. Instead, we focus on the conditional average treatment effect (CATE) as the target estimand,
\begin{align}
    \tau(x) = \mathbb{E}\!\left[ Y(1) - Y(0) \mid X = x \right],
\end{align}
which characterizes treatment effect heterogeneity across covariate profiles. Thus, throughout this paper, beneficiary identification is based on $\tau(x)$ evaluated at a candidate patient's baseline covariates, not on the unobservable unit-level causal contrast \(\tau_i\). A selected ``beneficiary" should be interpreted as a candidate whose covariate profile shows evidence of conditional treatment benefit. In practice, the CATE \(\tau(x)\) is often estimated using flexible modeling approaches, yielding an estimator \(\hat{\tau}(x)\), which we refer to as the predicted treatment effect, or equivalently a CATE-based individualized treatment effect prediction. A common strategy is to identify beneficiaries by thresholding or ranking these predictions; for example, one may select individuals with \(\hat{\tau}(X_i)>0\). However, such approaches typically ignore uncertainty in both estimation and selection, making it difficult to assess the reliability of individual-level treatment decisions.

Motivated by this limitation, we move beyond point estimation and formulate individualized treatment recommendation as a multiple hypothesis testing problem. Specifically, for each candidate individual with covariates $X_{n+m+j}$, we consider
\begin{equation}\label{eq:hypothesis}
H_{0j}: \tau(X_{n+m+j}) \le c_j
\quad \text{versus} \quad
H_{1j}: \tau(X_{n+m+j}) > c_j,
\end{equation}
for \(j=1,\dots,n_t\), where \(c_j\) represents a clinically meaningful threshold. Hence, the set of null hypotheses is denoted as $\mathcal{H}_0=\{j:\tau(X_{n+m+j})\le c_j\}$. This formulation forms the basis of our conformal inference framework, which enables uncertainty-aware identification of treatment beneficiaries.

\subsection{Proposed Framework}\label{sec:method_framework}
Figure \ref{fig:framework} provides an overview of the proposed framework, which integrates pseudo-outcome–based ITE estimation with conformal inference to identify treatment beneficiaries while controlling the FDR. The data sources and structure have been introduced in Section~\ref{sec:method_basicsetup}. We now present the methodological details of each remaining phase of the framework.

\begin{figure}[ht]
    \centering
    \includegraphics[width=0.8\linewidth]{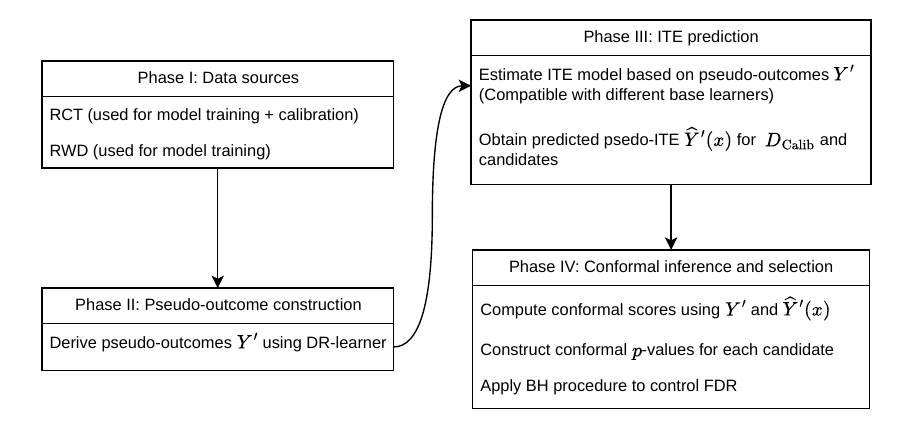}
    \caption{Flowchart of the proposed framework for identifying treatment beneficiaries}
    \label{fig:framework}
\end{figure}

\subsubsection{Pseudo ITE Outcome}\label{sec:method_Pseudo}
Because the ITE $\tau_i = Y_i(1) - Y_i(0)$ is unobservable, it cannot be used directly as a regression target. A common strategy is to construct a pseudo-outcome from observed data whose conditional expectation recovers the CATE \(\tau(x)\). In this paper, we adopt the doubly-robust (DR) pseudo-outcome construction following \citet{Kennedy2023}. This approach has been widely used in the literature on heterogeneous treatment effect estimation, which aims to converts the observed outcome $Y$, along with treatment assignment, propensity score, and outcome regression functions into, into a pseudo ITE outcome $Y'$ that can be used as a regression target for CATE learning. DR pseudo-outcome construction also allows flexible nuisance modeling and retains double robustness with respect to the propensity score and outcome regression models. 

Let \(e(x)=\Pr(A=1\mid X=x)\) denote the propensity score, and let $\mu_a(x)=\mathbb{E}(Y\mid A=a,X=x)$ denote the outcome regression function for treatment arm $a\in\{0,1\}$. These nuisance functions are defined under the identification assumptions introduced in Section~\ref{sec:method_basicsetup}, which ensure that the treatment effect \(\tau(x)\) is well-defined. Given estimates $\hat e(x)$, $\hat\mu_1(x)$, and $\hat\mu_0(x)$, we define the pseudo-outcome as
\begin{equation}\label{eq:pseudoITE}
Y'_i = \frac{A_i-\hat e(X_i)}{\hat e(X_i)\left(1-\hat e(X_i)\right)} \left(Y_i-\hat\mu_{A_i}(X_i)\right)+\hat\mu_1(X_i)-\hat\mu_0(X_i).
\end{equation}

The validity of the pseudo-outcome builds on the identification assumptions in Section~\ref{sec:method_basicsetup}, which ensure that \(\tau(x)\) is identifiable. In addition, standard regularity conditions on the nuisance functions are required. In particular, the propensity score must satisfy \(0 < e(x) < 1\) for all \(x\) in the support of \(X\), so that the weighting term in \eqref{eq:pseudoITE} is well-defined. At the population level, replacing the estimated nuisance functions with their true counterparts yields a pseudo-outcome whose conditional expectation satisfies
\begin{equation}\label{eq:pseudoITE_consistent}
    \mathbb{E}(Y'_i \mid X_i = x) = \tau(x).
\end{equation}
Moreover, the DR construction ensures that this property continues to hold if either the propensity score model or the outcome regression models are correctly specified \citep{Kennedy2023}. \eqref{eq:pseudoITE_consistent} is critical and useful because it simplifies heterogeneous treatment effect estimation to a standard regression problem. Once pseudo-outcomes are constructed, we can regress \(Y'\) on \(X\) using flexible machine learning methods to obtain an estimator \(\hat\tau(x)\). Therefore, $Y'$ should be viewed as a regression target for learning the CATE $\tau(x)$, rather than as an observed ITE.

In our framework, both the RWD and a subset of the RCT data are used to estimate nuisance functions and construct pseudo-outcomes, leveraging the larger sample size and broader covariate support of the RWD. Meanwhile, an independent subset of RCT data is reserved for conformal calibration, ensuring valid inference in subsequent steps. To reduce overfitting and maintain independence between model training and calibration, we employ sample splitting and cross-fitting \citep{Chernozhukov2018}. Specifically, the training data are partitioned into $K$ folds, and for each fold $k \in \{1,\dots,K\}$, nuisance functions are estimated using the remaining data and applied to the held-out fold to construct pseudo-outcomes. The resulting pseudo-outcomes are then pooled and used to train a regression model for \(\tau(x)\). This framework is flexible and accommodates a wide range of learning algorithms for both nuisance estimation and treatment effect modeling, including parametric, semiparametric, and modern machine learning methods. The implementation details are summarized in Algorithm~\ref{alg:phase2_pseudoITE}.

\begin{algorithm}[ht]
\caption{Construct Pseudo-ITEs (DR-Learner)}
\label{alg:phase2_pseudoITE}
\SetKwInOut{Input}{Input}
\SetKwInOut{Output}{Output}
\SetKwInput{KwData}{Input}
\SetKwInput{KwResult}{Output}
\KwData{$\mathcal{D}_{\mathrm{Train}}$ (comprising RWD and $\mathrm{RCT}_{\mathrm{Train}}$), $\mathcal{D}_{\mathrm{Calib}}$, number of folds $K$}
\KwResult{Pseudo-ITE outcomes $Y'$ for both $\mathcal{D}_{\mathrm{Train}}$ and $\mathcal{D}_{\mathrm{Calib}}$}

\vspace{1mm}
\tcc{1 Propensity Score Estimation}
For RCT samples in $\mathcal{D}_{\mathrm{Train}}$ and $\mathcal{D}_{\mathrm{Calib}}$, set $\hat{e}(x) = 0.5$\;
For RWD samples in $\mathcal{D}_{\mathrm{Train}}$, estimate $\hat{e}(x)$ using a selected propensity model\;

\vspace{1mm}
\tcc{2 Cross-fitted Outcome Modeling}
Randomly split $\mathcal{D}_{\mathrm{Train}}$ into $K$ approximately equal folds $\mathcal{I}_1,\dots,\mathcal{I}_K$\;
\For{$k=1,\dots,K$}{
    Let the training subset be $\mathcal{D}_{-k} = \mathcal{D}_{\mathrm{Train}}\setminus \mathcal{I}_k$\;
    Train outcome models on $\mathcal{D}_{-k}$ for $A=1$ and $A=0$ to obtain $\hat{\mu}^{(-k)}_1(x)$ and $\hat{\mu}^{(-k)}_0(x)$\;
    
    \For{$i \in \mathcal{I}_k$}{
        Obtain out-of-fold predictions $\hat{\mu}^{(-k)}_1(X_i)$ and $\hat{\mu}^{(-k)}_0(X_i)$\;
    }
    
    \For{$j \in \mathcal{D}_{\mathrm{Calib}}$}{
        Obtain out-of-sample predictions $\hat{\mu}^{(-k)}_1(X_j)$ and $\hat{\mu}^{(-k)}_0(X_j)$\;
    }
}
Average the $K$ predictions for the calibration set: $\bar{\mu}_a(X_j) = \frac{1}{K}\sum_{k=1}^K \hat{\mu}^{(-k)}_a(X_j)$\;

\vspace{1mm}
\tcc{3 Construct Pseudo-ITEs}
\For{each individual $i \in \mathcal{D}_{\mathrm{Train}} \cup \mathcal{D}_{\mathrm{Calib}}$}{
    \eIf{$i \in \mathcal{D}_{\mathrm{Train}}$ (with $i \in \mathcal{I}_k$)}{
        Set $\hat{\mu}_{A_i}(X_i)$, $\hat{\mu}_1(X_i)$, and $\hat{\mu}_0(X_i)$ to the out-of-fold predictions from model $k$\;
    }{
        Set $\hat{\mu}_{A_i}(X_i)$, $\hat{\mu}_1(X_i)$, and $\hat{\mu}_0(X_i)$ to the averaged predictions $\bar{\mu}$\;
    }
    
    Compute the doubly-robust pseudo-outcome $Y'$ based on Equation \eqref{eq:pseudoITE}\;
}
\Return Set of pseudo-outcomes $\{Y'_i\}$ for all $i \in \mathcal{D}_{\mathrm{Train}} \cup \mathcal{D}_{\mathrm{Calib}}$
\end{algorithm}

The DR pseudo-outcome is used as a regression target for learning the CATE. At the population level, when the nuisance functions are correctly specified, it satisfies $\mathbb{E}(Y' \mid X=x)=\tau(x)$. The DR pseudo-outcome introduced here is useful not only for treatment effect estimation, but also for subsequent conformal inference. In particular, conformal procedures based on DR pseudo-outcomes can yield valid inference for individual treatment effects under suitable stochastic ordering conditions between pseudo-outcome-based conformity scores and their oracle counterparts \citep{alaa2024conformal}. This provides the theoretical motivation for our conformal inference framework, which we describe in the next section.

\subsubsection{Conformal Prediction for Pseudo ITE}\label{sec:method_conformal}
We now use the pseudo-outcomes introduced in Section~\ref{sec:method_Pseudo} as the basis for conformal inference. After regressing \(Y'\) on covariates, the fitted function \(\hat{\tau}(x)\) gives a CATE-based prediction of treatment benefit for a candidate patient with covariates \(X=x\). Instead of selecting patients only by ranking or thresholding these predictions, we construct conformal scores from the pseudo-outcomes and use them to form conformal \(p\)-values for individual-level beneficiary selection.

\vspace{1em}
\noindent\textbf{Conformal scores.}
We start by defining a conformal score $V(x,y)$, which measures the compatibility of a candidate outcome value $y$ with the fitted treatment effect model at covariate value $x$. 
Common examples include absolute residuals, standardized residuals, and conformalized quantile regression scores \citep{shafer2008tutorial}. In our setting, \(y\) is taken to be the pseudo-outcome \(Y'\), and the fitted model is \(\hat{\tau}(x)\). Thus, \(V(x,Y')\) measures the discrepancy between the pseudo-outcome and the fitted CATE model, and these scores are used to construct conformal \(p\)-values for individual-level beneficiary selection.

For individuals in the training set $\mathcal{D}_{\mathrm{Train}}$ and the calibration set $\mathcal{D}_{\mathrm{Calib}}$, we compute the pseudo-outcome-based conformal scores as
\[
V_i = V(x=X_i, y=Y'_i), \quad i \in \mathcal{D}_{\mathrm{Train}} \cup \mathcal{D}_{\mathrm{Calib}}.
\]
For a candidate individual $j$, $j \in \{1,\dots,n_t\}$, for whom $Y'$ is not available, we instead evaluate the score at the null threshold:
\[
\hat V_j=V_{n+m+j} = V(x=X_{n+m+j}, y=c_j).
\]
Here, $\hat V_j$ is the null-imputed boundary score for candidate \(j\). Intuitively, the score $\hat V_j$ quantifies how far the candidate $j$ is from the threshold \(c_j\) specified under the null hypothesis. In this sense, it reflects how compatible the candidate is with being a non-beneficiary, i.e., whether $\tau(X_{n+m+j}) \le c_j$.

\begin{definition}[Monotone conformal score]
A score function $V(x,y)$ is defined to be monotone if for every fixed $x\in \mathcal{X}$,
\[
y_1 \le y_2 \quad \Longrightarrow \quad V(x,y_1) \le V(x,y_2).
\]
\end{definition}
Monotonicity aligns the score with the one-sided hypothesis \eqref{eq:hypothesis}. In particular, under the null hypothesis, 
\[
    \tau(X_{n+m+j})\le c_j\quad\Longrightarrow\quad V(X_{n+m+j},\tau(X_{n+m+j})) \le V(X_{n+m+j},c_j).
\]
This implication provides the key comparison between the oracle candidate score and the null-imputed boundary score used in the conformal \(p\)-value construction. 

For simplicity, here, we take the clinically meaningful threshold to be common across candidates, i.e., \(c_j\equiv c\). 
In this paper, we use the following one-sided score function \citep{Jin2023SelectionbyConP}:
\begin{align}\label{eq:conformal_score}
    V(x,y)=M\cdot \mathbb{I}\{y>c\}-\hat{\tau}(x)
\end{align}
where $M$ is chosen to be sufficiently large so that
$M > 2\sup_{x\in\mathcal X_{\rm target}}|\hat\tau(x)|$, with $\mathcal X_{\rm target}$ denoting the covariate region of the calibration and candidate populations. When \(c_j\equiv 0\), the candidate score simplifies to
\[\hat V_j = -\hat\tau(X_{n+m+j}),\]
while for calibration units, 
\[V_i = M \mathbb{I}\{Y'_i > 0\} - \hat\tau(X_i).\]
This choice ensures that smaller values of candidate conformal score $\hat V_j$ correspond to stronger evidence against the null. More generally, other monotone score functions can also be used, provided their direction is aligned with the hypothesis being tested.

One appealing feature of conformal inference is that it uses the calibration scores as an empirical reference distribution, rather than relying on large-sample approximations or a fully specified parametric null distribution like asymptotic inference \citep{shafer2008tutorial}. 
In our setting, the calibration set is an independent subset of RCT subjects,
\[\mathcal{D}_{\mathrm{Calib}} = \{(Y_i, A_i, X_i)\}_{i=1}^{n_c} \subset \mathcal{D}_{\mathrm{RCT}}, n_c < n.\]

For the one-sided score in \eqref{eq:conformal_score}, smaller values of \(\hat V_j\) correspond to stronger evidence against \(H_{0j}: \tau(X_{n+m+j}) \le c_j\). In particular, when \(c_j=0\), \(\hat V_j=-\hat\tau(X_{n+m+j})\), so a smaller \(\hat V_j\) indicates a larger predicted treatment benefit. Thus, a candidate whose conformal score is unusually small relative to the calibration scores provides stronger evidence of conditional treatment benefit and is more likely to be selected as a beneficiary. The resulting rank comparison is then converted into a conformal \(p\)-value in the next paragraph.

\vspace{1em}
\noindent\textbf{Conformal $p$-values.}
Building on this calibration step, we quantify the extremeness of the candidate score $V_{n+m+j}$ by comparing it to the empirical distribution of calibration scores. The calibration scores' empirical distribution serves as the null distribution in conformal inference, which avoids the requirement of large sample theory and model specifications in asymptotic inference to approximate null distribution. Specifically, for candidate $j$, define $\hat V_j = V(X_{n+m+j},c_j)$ as the conformal score evaluated at the null threshold. Following the lower-tail direction of the one-sided score in \eqref{eq:conformal_score}, the conformal $p$-value is defined as
\begin{align}\label{eq:pvalue_orig}
    p_j=\frac{\sum_{i=1}^{n_c}\mathbb{I}(V_i<\hat V_j)+U_j\cdot (1+\sum_{i=1}^{n_c}\mathbb{I}(V_i=\hat V_j))}{n_c+1},
\end{align}
where $U_j \sim \mathrm{Unif}(0,1)$ is used to break ties. 
This rank-based conformal $p$-value measures how unusually small the candidate's null-imputed conformal score is relative to the calibration scores. In other words, smaller values of $p_j$ indicate stronger evidence against $H_{0j}:\tau(X_{n+m+j})\le c_j$ and therefore stronger evidence that candidate $j$ is likely to benefit from treatment. \eqref{eq:pvalue_orig} is used in simulation and case study to obtain the conformal $p$-values. 

Because we construct conformal scores and conformal \(p\)-values using \(Y'\) rather  than the original outcome \(Y\), it requires additional justification. 
Validity for the pseudo-outcome $Y'$ does not, by itself, imply validity for inference on the individual treatment effect. To address this issue, we build on the stochastic-ordering framework \citep{alaa2024conformal}. Their result shows that, if the conformity scores based on pseudo-outcomes are conservative relative to the corresponding oracle scores that would have been computed using the target CATE $\tau(X)$, then conformal inference based on pseudo-outcomes remains valid for CATE-level inference. Their result is stated in terms of coverage for conformal prediction intervals, whereas our procedure is formulated through conformal $p$-values for one-sided testing. Therefore, in Supplementary Material~\ref{app:coverage_to_pvalue}, we show how a valid one-sided conformal prediction set induces a generalized valid conformal $p$-value satisfying
\begin{equation*}
    \mathbb P\{p_j\le \alpha,\ H_{0j}\text{ is true}\}\le \alpha, \qquad \alpha\in[0,1],
\end{equation*}
which provides the testing interpretation needed for candidate-level screening or testing. 

The validity of conformal rank-based inference relies on an exchangeability assumption for the scores used in the conformal comparison \eqref{eq:pvalue_orig}. 

\begin{assumption}[Exchangeability for candidate-level conformal validity]
\label{as:exchangeability}
For candidate $j$, let
\[
    V_j^*=V(X_{n+m+j},\tau(X_{n+m+j}))
\]
denote the oracle candidate score. Under the null hypothesis $H_{0j}:\tau(X_{n+m+j})\le c_j$, the oracle candidate score $V_j^*$ is exchangeable with the calibration scores $\{V_i:i\in\mathcal D_{\mathrm{Calib}}\}$.
\end{assumption}
Assumption \ref{as:exchangeability} requires that, for a single candidate, the calibration scores and the candidate's oracle score $V_j^*$ to be exchangeable under the null. Under Assumption~\ref{as:exchangeability}, the rank of $V_j^*$ among $\{V_i:i\in\mathcal D_{\mathrm{Calib}}\}\cup\{V_j^*\}$ is uniformly distributed. This establishes the finite-sample conformal rank-validity underlying the candidate-level conformal $p$-value. In practice, however, the proposed procedure is applied simultaneously to a group of candidates, so multiplicity must also be addressed. We discuss FDR control for simultaneous candidate selection in Section~\ref{sec:fdr_ctr}.

\begin{remark}[Connection to conformal intervals]
\label{rem:interpretation}
The conformal $p$-value also introduces an equivalent interpretation in terms of one-sided conformal intervals. In particular, the split conformal procedure \citep{lei2018distribution} yields a lower
prediction set of the form
\begin{align*}
    \hat C(X_{n+m+j},1-\alpha) = \left[\eta(X_{n+m+j},1-\alpha), +\infty\right),
\end{align*}
where $\eta(X_{n+m+j},1-\alpha)$ denotes the $(1-\alpha)$ empirical quantile of the calibration scores. Then the one-sided test of $H_{0j}:\tau(X_{n+m+j})\le c_j$ rejects when $\eta(X_{n+m+j},1-\alpha)> c_j$.
The conformal $p$-value can then be written as
\begin{equation*}
p_j=\inf \Big\{ \alpha \in (0,1) : c_j < \eta(X_{n+m+j},1-\alpha) \Big\}.
\end{equation*}
In other words, $p_j$ is the smallest nominal level at which the null hypothesis is rejected. 
\end{remark}

Furthermore, the proposed framework can be extended naturally to other outcome types, such as binary outcomes, by using appropriate generalized models for $\hat\tau(x)$, e.g., logistic regression, and corresponding conformal score functions.

\subsubsection{False discovery rate (FDR) Control}\label{sec:fdr_ctr}
After obtaining the individual-level conformal $p$-value, we can test the individual-level hypothesis defined in \eqref{eq:hypothesis}. For example, for a pre-specified significance level $\alpha_j$, we reject $H_{0j}$ if $p_j \le \alpha_j$. In practice, however, the goal is rarely to test a single candidate in isolation. Instead, we often aim to screen a group of candidate individuals simultaneously and identify those most likely to benefit from treatment. This naturally leads to a multiple testing problem. Although each individual test has a controlled probability of false rejection, false discoveries can accumulate when many hypotheses are tested simultaneously \citep{Ranganathan2016}. If multiplicity is not addressed, the resulting selected beneficial subgroup may contain an inflated number of non-beneficial candidates, which is not ideal. Consequently, this inflation can compromise the reliability of downstream clinical or trial-design decisions.

To address multiplicity issue, there are two commonly used approaches, including family-wise error rate (FWER) control and FDR control. FWER controls the probability of making at least one false discovery \citep{Tukey1953,Benjamini2002}, whereas FDR controls the expected proportion of false discoveries among all discoveries \citep{benjamini1995}. In that sense, FDR is generally less stringent than FWER and is often more appropriate for large-scale screening problems. This is because for large-scale screening, the goal is to identify a subset of promising candidates rather than to avoid any false selection. In our setting, the objective is to select a subgroup of candidates who are likely to benefit from treatment, for example to inform a future follow-up study or a second-stage enrichment trial. Therefore, the control of FDR is more aligned with the setting of interest and is the primary focus in this paper.

Let
\[
    \mathcal H_0=\{j:\tau(X_{n+m+j})\le c_j\}
\]
denote the set of true null hypotheses, and let $\mathcal R\subseteq\{1,\ldots,n_t\}$ denote the selected set (discovery set). The FDR is defined as 
\begin{equation}\label{eq:fdr_def}
    \mathrm{FDR}=\mathbb E\left[\frac{\sum_{j=1}^{n_t}\mathbb{I}\{j\in\mathcal R,\ j\in\mathcal H_0\}}{1\vee |\mathcal R|}\right],
\end{equation}
where $1\vee |\mathcal R|=\max\{1,|\mathcal R|\}$. In healthcare and trial-design applications, controlling this quantity is very important. For example, when an aggressive treatment is under study, then false selections may expose non-beneficial individuals to unnecessary overtreatment. By controlling the FDR at a pre-specified level $q$, the proposed procedure ensures that the expected proportion of false selections among all selected candidates is bounded by $q$, i.e., $\mathrm{FDR}\le q$.

To control FDR, we apply the Benjamini-Hochberg (BH) procedure to the conformal $p$-values $\{p_j:j=1,\ldots,n_t\}$ for candidate group \citep{benjamini1995}. An alternative multiplicity adjustment is the Bonferroni correction, which selects candidates who satisfy $p_j\le q/n_t$. The Bonferroni correction controls the probability of making at least one false selection by allocating the overall error tolerance equally across all candidates. Therefore, although it also provides valid error control, it is often considered conservative for FDR-oriented screening problems and can substantially reduce power compared to BH procedure \citep{Jin2023SelectionbyConP}.

\begin{assumption}[Conditional exchangeability of conformal scores]
\label{as:score_exchangeability}
Conditional on the training stage, including the fitted nuisance functions,
the pseudo-outcome construction, and the fitted CATE model $\hat\tau(\cdot)$,
for each candidate $j$, the conformal scores 
\[\{V_i:i\in\mathcal D_{\mathrm{Calib}}\}\cup\{V_j^*\}\]
are exchangeable conditional on the candidate boundary scores
\[
    \{\hat V_\ell: \ell\ne j\}.
\]
In addition, the scores
\[
    \{V_i: i\in\mathcal D_{\mathrm{Calib}}\}
    \cup
    \{\hat V_\ell: \ell=1,\ldots,n_t\}
\]
have no ties almost surely.
\end{assumption}

Assumption~\ref{as:score_exchangeability} strengthens the candidate-level exchangeability condition in Assumption~\ref{as:exchangeability}. It is used for the finite-sample FDR argument because, under BH, whether candidate $j$ is selected depends not only on $p_j$, but also on the conformal $p$-values of the other candidates through the data-adaptive BH threshold. Since the calibration scores are constructed from pseudo-outcomes rather than oracle CATE values, the stochastic-ordering condition discussed in Section~\ref{sec:method_conformal} is also required. In the simulations and case study, we use the randomized conformal $p$-value in \eqref{eq:pvalue_orig}, where $U_j$ is used to break ties. For the finite-sample FDR control, we state the Theorem \ref{thm:fdr} using the deterministic conformal $p$-value, following the deterministic formulation of \citet{Jin2023SelectionbyConP}. This avoids auxiliary tie-breaking randomness in the proof and is stated under the no-ties assumption.

\begin{theorem}\label{thm:fdr}
Suppose that the score function $V(x,y)$ is monotone in $y$, the pseudo-outcome-based scores are conservative relative to the corresponding oracle CATE-based scores, and suppose that Assumption~\ref{as:score_exchangeability} holds, i.e., the conformal scores $\{V_1,\ldots,V_{n_c},V_j^*\}$ are exchangeable conditional on $\{\hat V_\ell:\ell\ne j\}$. Let $\mathcal H_0=\{j:\tau(X_{n+m+j})\le c_j\}$ denote the set of true null hypotheses. Assume there is no tie almost surely among $\{V_1,\ldots,V_{n_c}\} \cup \{\hat V_{\ell}:\ell=1,\ldots,n_t\}$. 
Apply the BH procedure at level $q\in(0,1)$ to the the deterministic conformal $p$-values
    \begin{align*}
        p_j=\frac{1+\sum_{i=1}^{n_c}\mathbb{I}\{V_i<\hat V_j\}}{n_c+1}.
    \end{align*}
Let the resulting selected set be
\[
    \mathcal R=\{j:p_j\le q\hat r/n_t\},
    \qquad
    \hat r=
    \max\left\{
    r:\sum_{j=1}^{n_t}\mathbb{I}\{p_j\le qr/n_t\}\ge r
    \right\}.
\]
Then, the FDR is well controlled as
\[
    \mathrm{FDR}=\mathbb E\left[\frac{\sum_{j=1}^{n_t}\mathbb{I}\{j\in\mathcal R,\ j\in\mathcal H_0\}}{1\vee |\mathcal R|}\right]\le q .
\]
\end{theorem}
Theorem~\ref{thm:fdr} shows that applying the BH procedure to the conformal $p$-values can control the FDR at the target level $q$. In our setting, this means that the expected proportion of false selections among selected candidates is bounded by $q$, which provides a multiplicity-adjusted validity guarantee for simultaneous screening. The proof is provided in Supplementary Material~\ref{app:FDR_control}. 

In the proposed framework, external data are used only to improve the training of the pseudo-ITE model, but not to define the calibration reference distribution for conformal ranking. Therefore, intuitively, the external data help the model learn treatment-effect patterns from a larger and more diverse sample, while the held-out RCT calibration set determines how strong the evidence must be before a candidate is selected. Because the external data are not used for conformal calibration, the potential hidden bias or population shift in the external data does not directly affect the FDR guarantee, as long as the RCT calibration set remains exchangeable with the candidate population. However, if the external data are poorly aligned with the trial population, they may still affect the fitted CATE model. This is consistent with prior work on prognostic adjustment, where historical data are used to improve prediction or efficiency while inference remains anchored in the randomized trial data \citep{HojbjerreFrandsen2026}. This statement is also supported empirically in the simulation studies in Section~\ref{sec:sim_outcomedrift}. To assist in the implementation of BH procedure, we provide Algorithm \ref{alg:cfBH}. In practice, people can apply the BH procedure to the conformal \(p\)-values and select candidates whose BH-adjusted conformal \(p\)-values are no larger than the target FDR level \(q\). The FDR tolerance level $q$ should be pre-specified, and should be chosen based on the clinical or study-design tolerance for false beneficiary selections. 
In screening-oriented or enrichment settings, moderately liberal values such as \(q=0.10\) or \(q=0.20\) are often considered to balance discovery and false-selection control, whereas more stringent values such as \(q=0.05\) may be preferred when selected candidates directly inform high-stakes treatment decisions.

\begin{algorithm}[ht]
\caption{BH procedure: Selection by prediction with conformal $p$-values}
\label{alg:cfBH}
\SetKwInOut{Input}{Input}
\SetKwInOut{Output}{Output}
\SetKwInput{KwData}{Input}
\SetKwInput{KwResult}{Output}

\KwData{Calibration set $\mathcal D_{\mathrm{Calib}}$ with pseudo-outcomes $\{(X_i,Y_i') : i=1,\ldots,n_c\}$; candidate covariates $\{X_{n+m+j}:j=1,\ldots,n_t\}$; individual-level thresholds $\{c_j:j=1,\ldots,n_t\}$; target FDR level $q\in(0,1)$; monotone conformal score $V:\mathcal X\times\mathcal Y\to\mathbb R$}
\KwResult{Selected candidate set $\mathcal R$ after BH adjustment}

\vspace{1mm}
\tcc{1 Compute calibration and candidate conformal scores}
Compute calibration scores: $V_i=V(X_i,Y_i'), \qquad i=1,\ldots,n_c.$

Compute candidate boundary scores: $\hat V_j=V(X_{n+m+j},c_j), \qquad j=1,\ldots,n_t.$

\vspace{1mm}
\tcc{2 Construct conformal $p$-values}
For each candidate $j=1,\ldots,n_t$, compute the conformal $p$-value $p_j$ using \eqref{eq:pvalue_orig}\;

\vspace{1mm}
\tcc{3 Apply the BH procedure}
Compute
\[
    \hat r=\max\left\{r\in\{1,\ldots,n_t\}:\sum_{j=1}^{n_t}\mathbb I\{p_j\le qr/n_t\}\ge r \right\},
\]
with $\hat r=0$ if the set is empty\;

\Return BH-adjusted selection set $\mathcal R=\left\{j:p_j\le q\hat r/n_t\right\}$ 
\end{algorithm}

\section{Simulation}\label{sec:sim}
\subsection{Simulation Setup}\label{sec:sim_setup}
We conduct an extensive simulation study to evaluate the performance of the proposed conformal selection framework. The primary goal of the simulation is to assess whether the proposed method achieves finite-sample FDR control, which constitutes the main theoretical guarantee of our approach and is formally established in the Supplementary Material~\ref{app:FDR_control}. Conditional on valid FDR control, we further examine the method's ability to identify true treatment beneficiaries. This is quantified by power, defined as the expected proportion of true beneficiaries that are selected, where true beneficiaries are candidates satisfying \(\tau(X_{n+m+j})>c_j\), with \(c_j\equiv 0\) in our simulations. Accordingly, FDR and power are treated as the primary performance metrics throughout the simulation study.

To investigate performance across a range of data regimes, we consider RCT sample sizes of $n_{\text{RCT}}\in\{50,100,200\}$, where $n_{\text{RCT}}=50$ represents a small-sample setting motivated by early-phase trials or rare disease studies. To assess the impact of borrowing external data for effect model training and increasing covariate diversity, we generate an external dataset of size $n_{\text{RWD}}=800$. The candidate population to be tested is fixed at $n_{\text{Test}}=100$ in all scenarios.

\subsubsection{Data-generating process.}
We consider both linear and nonlinear CATE structures to evaluate method's robustness to model misspecification. For each unit $i$, covariates $\boldsymbol{X}_i=(X_{i1},X_{i2})^\top\in\mathbb{R}^2$ are independently generated with $X_{i1},X_{i2}\overset{\text{i.i.d.}}{\sim}\mathrm{Unif}(-1,1)$. For RCT units ($S=1$), treatment is randomly assigned according to $A_i\sim\mathrm{Bernoulli}(0.5)$, whereas for external units ($S=0$), we fix the treatment assignment probability at $\Pr(A_i=1)=0.4$, which is to be estimated.

The baseline mean outcome is specified as $\mu_0(\boldsymbol{X}_i)=0.5 + X_{i1} + X_{i2}$. Under the linear CATE scenario, the treatment effect is given by $\tau(\boldsymbol{X}_i)=c_0 + \beta_1 X_{i1} + \beta_2 X_{i2} + X_{i1}X_{i2}$, whereas under the nonlinear scenario, $\tau(\boldsymbol{X}_i)=c_0 + \gamma_1 X_{i1}X_{i2} + \gamma_2 \exp\{1.5(X_{i2}-1)\}$. The treated mean outcome is $\mu_1(\boldsymbol{X}_i)=\mu_0(\boldsymbol{X}_i)+\tau(\boldsymbol{X}_i)$, and potential outcomes are generated as $Y_i(0)=\mu_0(\boldsymbol{X}_i)+\varepsilon_{i0}$ and  $Y_i(1)=\mu_1(\boldsymbol{X}_i)+\varepsilon_{i1}$, where $\varepsilon_{ia}\sim\mathcal{N}(0,\sigma_i^2)$ independently for $a\in\{0,1\}$. We consider three noise structures, including two homoskedastic settings with $\sigma_i=\sigma\in\{0.3,0.6\}$, and one heteroskedastic setting in which $\sigma_i=\sigma_i(\boldsymbol{X}_i)$ depends on $\tau(\boldsymbol{X}_i)$. The candidate dataset follows the same data-generating mechanism as the RCT, which ensures the exchangeability between the calibration and testing sets, a key requirement for conformal validity.

\subsubsection{Introduce covariate shift and hidden bias.}
To assess robustness under covariate shift, we introduce two levels of severity, i.e., slight and moderate, corresponding to distributional differences between the RCT and external data. Specifically, for $S=0$, we first generate a super-population of size $N_{\text{super}}=10{,}000$ under the same covariate-generating mechanism, and then sample individuals according to a logistic selection model, $\Pr(S_i=1\mid \boldsymbol{X}_i)=\mathrm{logit}^{-1}(a_0+\eta^\top \boldsymbol{X}_i)$, where the parameter vector $\eta$ controls the degree of covariate shift and $a_0$ is chosen such that $\mathbb{E}\{\Pr(S=1\mid \boldsymbol{X})\}=(N_{\text{super}}-n)/N_{\text{super}}$.

We initially study the performance of proposed methods under covariate shift, assuming no hidden bias from the external data. In this setting, the CATE function and potential outcome distributions are shared across data sources, so the external data can improve beneficiary identification by increasing sample size and covariate-space diversity without introducing outcome drift. Then, we consider a more realistic setting in which hidden outcome bias may exist in the external data. To examine robustness to such hidden bias, we randomly select $50\%$ of external units to be biased. For these biased external subjects, we shift the control potential outcome mean by a constant $b\in\{0.4,0.8\}$, while the remaining external units follow the same potential-outcome models as the RCT. Because external data are used exclusively to augment the training set and are never included in the conformal calibration step, our framework is expected to prevent the RCT from being dominated by real-world data and provide a valid error rate control while compromise some power. 

\subsubsection{Base-modeling and error rate control strategies.}

Because our proposed conformal selection framework is model-agnostic, it can be integrated with a wide range of treatment effect learners. Therefore, in the simulation study, we evaluate its performance across diverse base learners rather than restricting attention to a single modeling strategy. Specifically, we consider five learners that span both frequentist and Bayesian approaches: linear regression; random forests implemented via the \texttt{ranger} package \citep{Wright2017ranger}; Super Learner implemented using the \texttt{SuperLearner} package \citep{vanDerLaan2007SuperLearner}; Bayesian additive regression trees (BART) implemented with the \texttt{BART} package \citep{Chipman2010BART,BART_pack}; and Bayesian model averaging (BMA) implemented using the \texttt{BAS} package \citep{Hoeting1999BMA,Clyde2011BAS}. Each learner is implemented both with and without incorporating external data for training. This design allows us to evaluate how external data affect the efficiency and power of the proposed framework, while maintaining RCT-based conformal calibration for inference. For decision-making, each modeling strategy is combined with three multiplicity-control approaches: no FDR control, Bonferroni correction, and the Benjamini-Hochberg (BH) procedure. Here, we prespecify the target FDR level to \(q=0.15\), reflecting a moderately stringent level appropriate for screening-oriented beneficiary selection.

All simulations are conducted with $1{,}000$ Monte Carlo replications. In each replication, the RCT data are randomly partitioned into a training set $\mathcal{R}_{\text{Train}}$ and a calibration set $\mathcal{D}_{\text{Calib}}$ using a $3{:}2$ split. When external data are not used, the training set is $\mathcal{D}_{\text{Train}}=\mathcal{R}_{\text{Train}}$; otherwise, $\mathcal{D}_{\text{Train}}=\mathcal{R}_{\text{Train}}\cup\mathcal{D}_{\rm RWD}$.

\subsection{Simulation Results}\label{sec:sim_res}
After describing the simulation setup, we now summarize and interpret the main findings from the simulation studies. We focus on FDR control and statistical power under different settings described above, starting with scenarios that involve covariate shift only and then followed by outcome drift.
\subsubsection{Covariate shift}\label{sec:sim_covdrift}
\begin{figure}[ht]
    \centering
    \includegraphics[width=\linewidth]{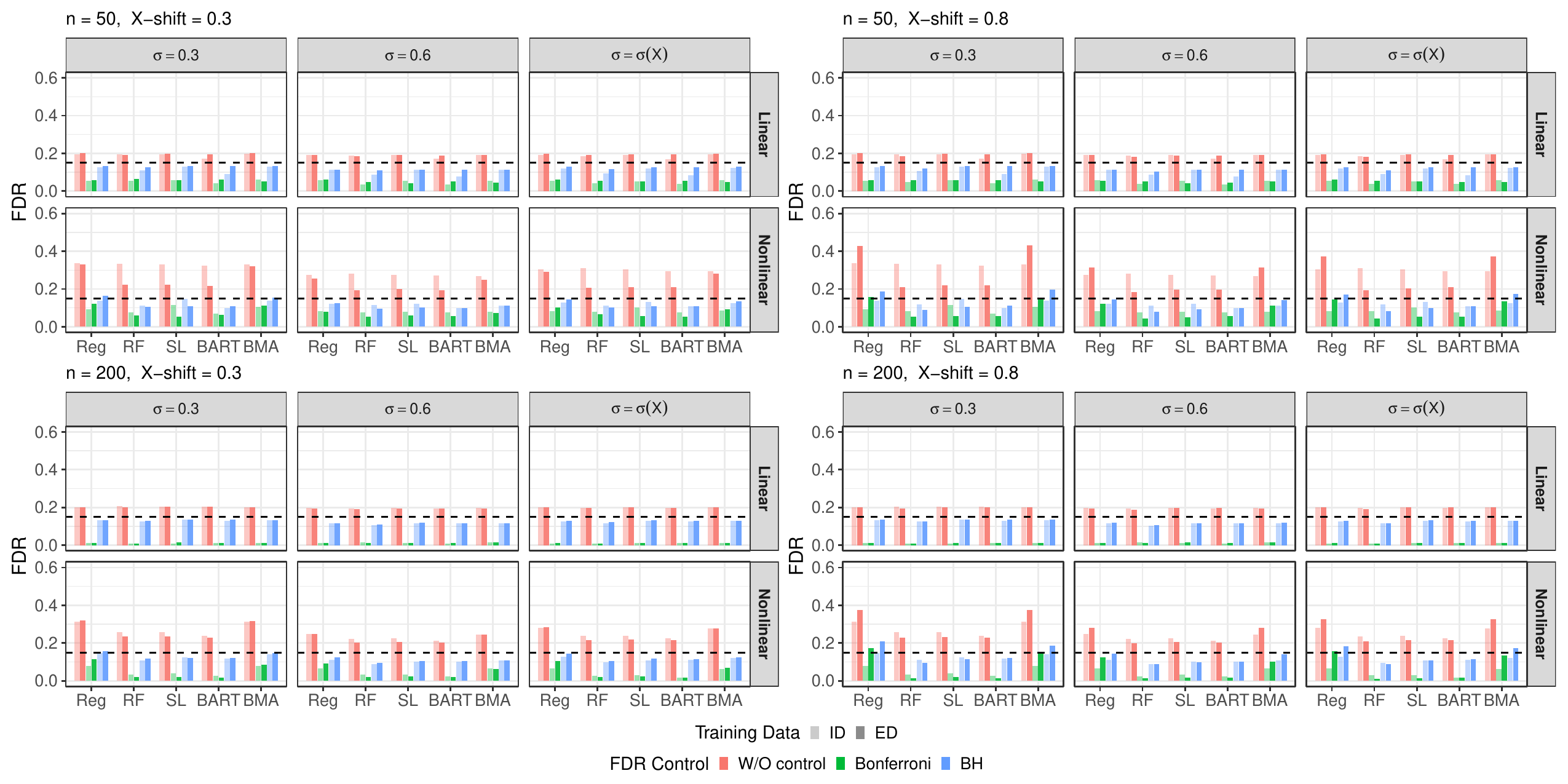}
    \caption{FDR across different scenarios given covariate shift}
    \label{fig:FDR_all_Xshift}
\end{figure}

Figure~\ref{fig:FDR_all_Xshift} shows the empirical FDR across different sample sizes, covariate shift levels, outcome model structures, and noise specifications. Across all scenarios, the BH procedure consistently controls the FDR at or below the pre-specified tolerance level of $q=0.15$. This holds even when the covariate shift is moderate to severe and when flexible machine learning methods are used to estimate individual treatment effects. In contrast, the Bonferroni procedure is uniformly conservative, leading to FDR values close to zero in nearly all settings. While this behavior guarantees error control, it comes at the cost of substantially reduced power, as shown in Figure~\ref{fig:Power_all_Xshift}. When without controlling FDR, noticeable FDR inflation can appear in several scenarios, particularly when the outcome model is nonlinear. The nonlinear structure might be more realistic and more challenging to model correctly. Therefore, it becomes very important to explicitly control FDR when the outcome-generating mechanism is complex and not directly testable in practice.

\begin{figure}[ht]
    \centering
    \includegraphics[width=\linewidth]{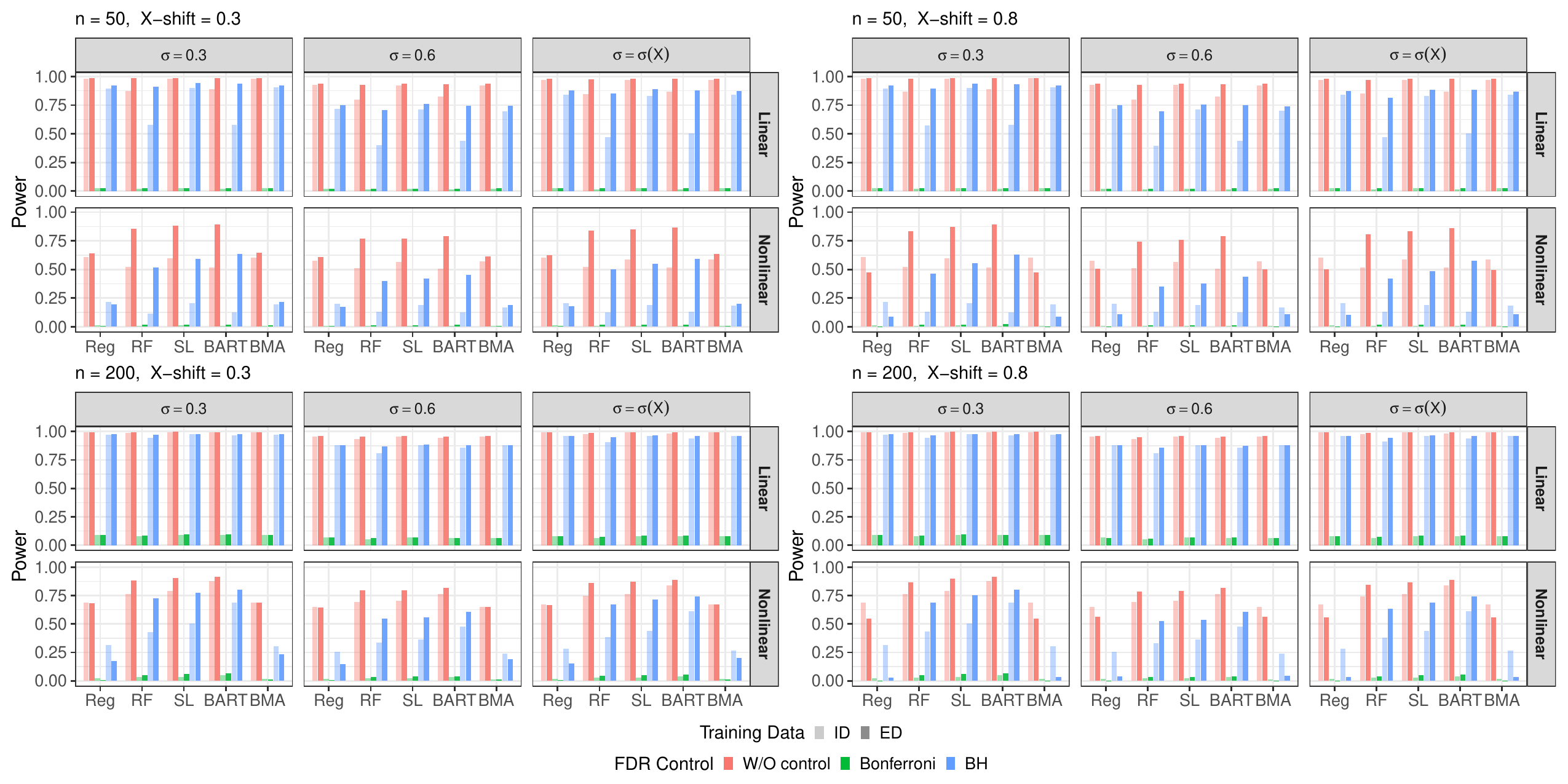}
    \caption{Power across different scenarios given covariate shift}
    \label{fig:Power_all_Xshift}
\end{figure}

Figure~\ref{fig:Power_all_Xshift} compares power across methods that rely only on RCT data and those that additionally borrow external data. Overall, borrowing external data leads to clear power gains compared to using RCT only, even in the presence of moderate to severe covariate shift. These power gains arise from the increased diversity of the covariate space, which improves the training of individual treatment effect models by providing more comprehensive covariate information. Nevertheless, the magnitude of the power improvement depends on both the degree of covariate shift and the RCT sample size. When the shift between RCT and external data is smaller, the borrowed data more closely resemble the target population, resulting in slightly larger power gains. More importantly, the benefit of borrowing external data is most pronounced when the RCT sample size is small. In these settings, external data play a essential role in stabilizing model estimation and improving the ability to identify treatment beneficiaries. Meanwhile, borrowing external data does not guarantee high power in all scenarios. When the baseline power using RCT data alone is very low, additional external data may provide limited benefit. This behavior is most evident for linear regression and Bayesian model averaging, which struggle to accurately capture complex outcome structures. As a result, their power remains relatively low regardless of whether external data are incorporated.

Within each combination of sample size and covariate shift, power varies substantially across base learners. Under a linear outcome model with low noise ($\sigma=0.3$), linear regression and BMA perform comparably to more flexible methods such as Super Learner and BART, achieving power close to $0.9$ in favorable settings (e.g., $n=50$, $\eta=0.3$). However, this advantage disappears when the outcome model becomes nonlinear. In these cases, Super Learner and BART consistently outperform simpler models, reflecting their ability to adapt to complex and unknown outcome structures without relying on strong modeling assumptions. Increasing the noise level from $\sigma=0.3$ to $\sigma=0.6$ leads to a general decrease in power across all methods, as larger noise makes it harder to distinguish true treatment beneficiaries. Interestingly, covariate-dependent noise does not uniformly worsen performance compared to constantly large noise structure ($\sigma=0.6$). Instead, it affects more on linear regression and BMA under nonlinear outcome models, while not impact that much for random forest, Super Learner, and BART. Finally, across all scenarios, the Bonferroni procedure yields substantially lower power compared to BH, which results from the trade-off between conservative error control and the ability to detect meaningful treatment effects. Furthermore, we empirically assess this stochastic-ordering condition under this simulation settings in Supplementary Material~\ref{app:empirical_stoch_main}.

\subsubsection{Outcome drift}\label{sec:sim_outcomedrift}
As shown in Figure \ref{fig:FDR_all_Yshift_0.8}, across all scenarios, the BH procedure continues to control the FDR at the pre-specified level of $q=0.15$. This result aligns with the core design principle of our framework, which separates model training from hypothesis calibration. In particular, the calibration set is constructed exclusively from RCT data, providing a valid reference distribution that targets the population of interest. Therefore, the conformal $p$-values remain valid under the exchangeability assumption between the testing and calibration sets even if the external data contains some hidden bias. Importantly, although external data may contain hidden bias due to outcome drift, incorporating such data during training does not compromise the validity of the hypothesis testing procedure nor inflate the FDR.

\begin{figure}[ht]
    \centering
    \includegraphics[width=\linewidth]{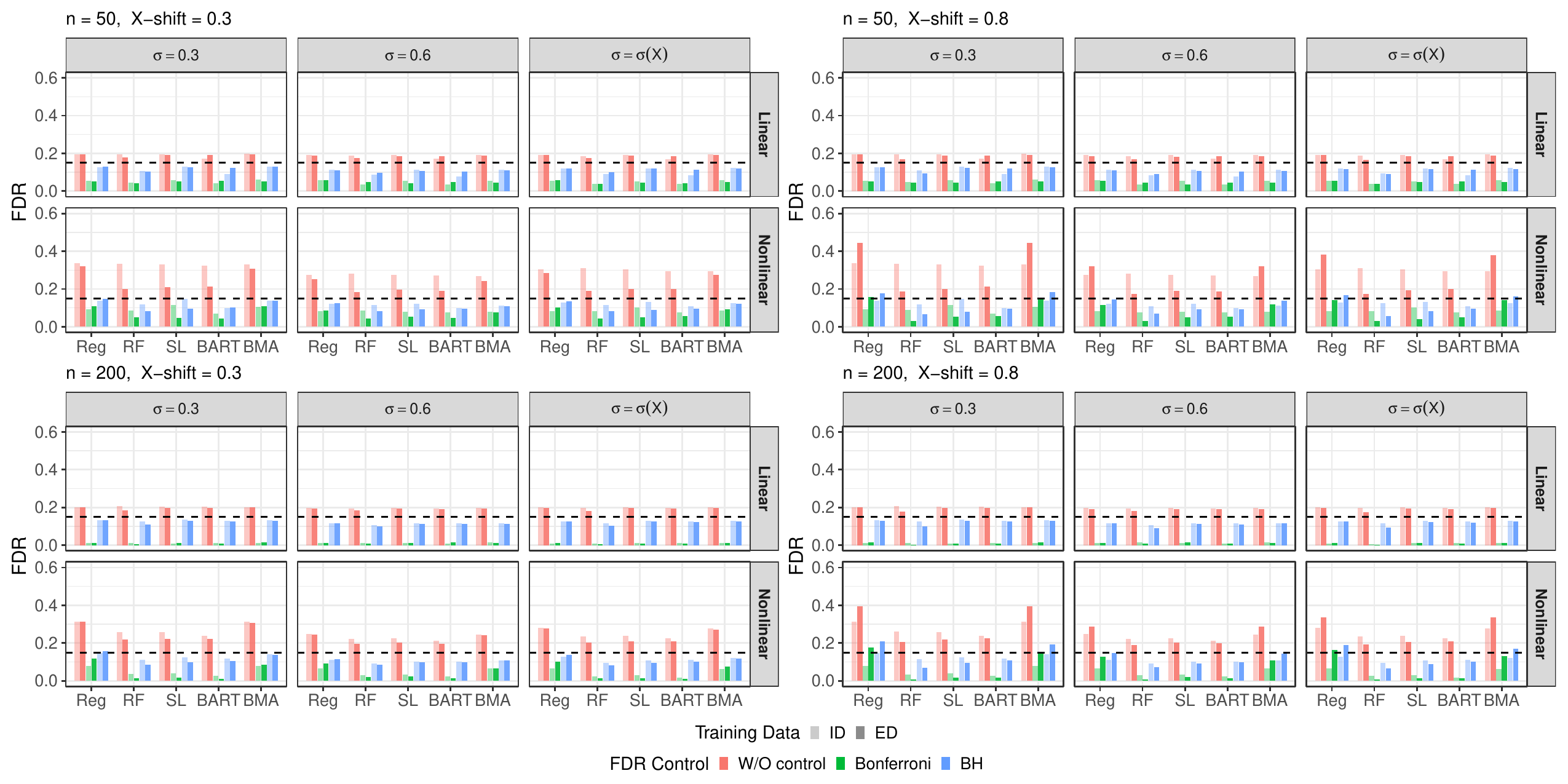}
    \caption{FDR across different scenarios given outcome drift ($b=0.8$)}
    \label{fig:FDR_all_Yshift_0.8}
\end{figure}

Correspondingly, Figure~\ref{fig:Power_all_Yshift_0.8} shows the power results. Borrowing external data during training continues to provide power gains, even in the presence of hidden bias. The relative performance of different base learners remains largely consistent with the covariate-shift-only setting. However, the overall magnitude of the power gain is slightly reduced compared to scenarios without outcome drift. This reduction is expected, as outcome drift introduces additional noise when external data are used for training, due to discrepancies in outcome distributions between the RCT and external data. Larger differences between the RCT and external data therefore limit the improvement in individual treatment effect estimation and beneficiary identification.

\begin{figure}[ht]
    \centering
    \includegraphics[width=\linewidth]{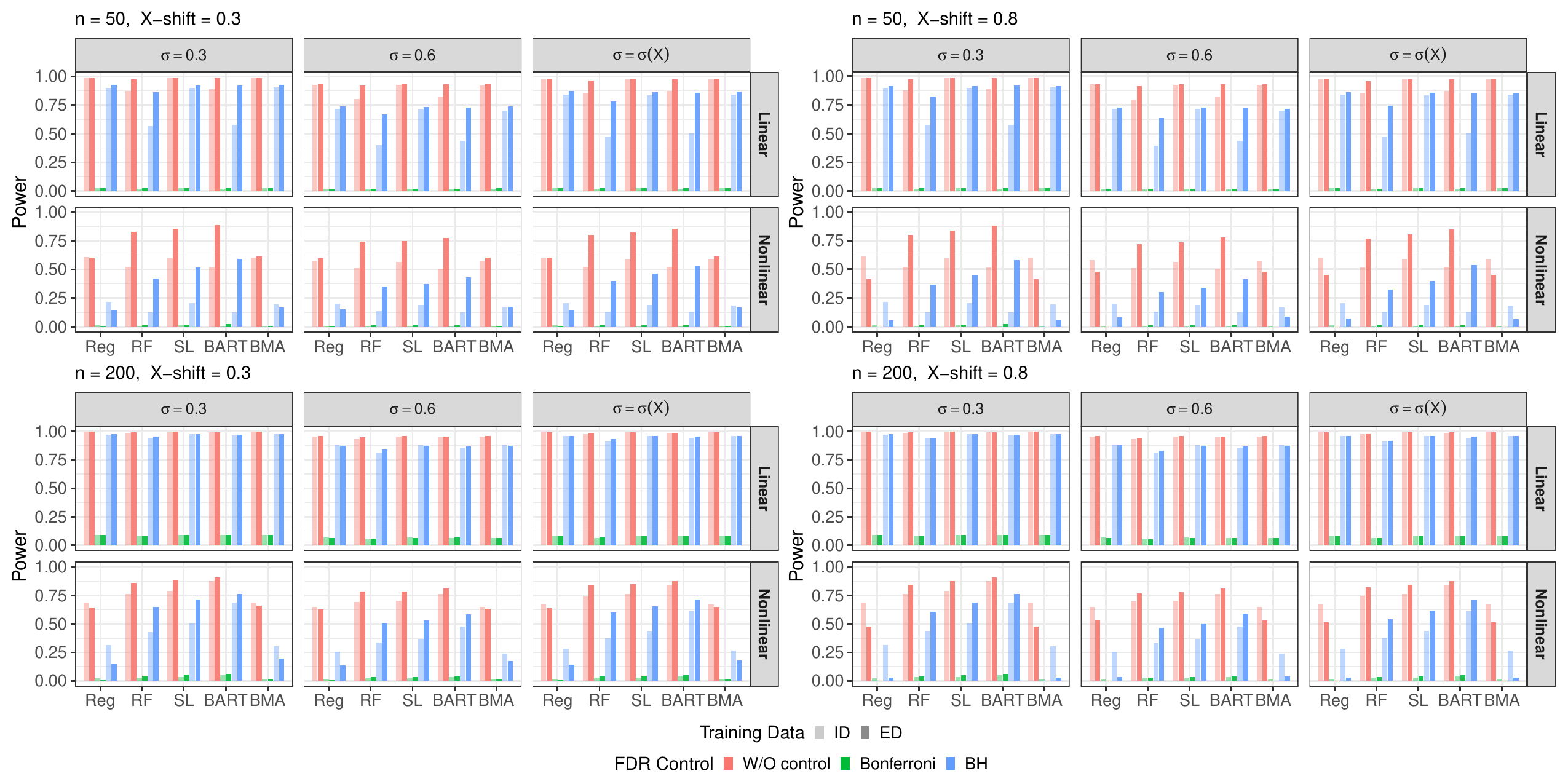}
    \caption{Power across different scenarios given outcome drift ($b=0.8$)}
    \label{fig:Power_all_Yshift_0.8}
\end{figure}

Additional simulation results with a smaller outcome drift magnitude ($b=0.4$) are provided in the Supplementary Material \ref{app:case_results}. These results lead to the same qualitative conclusions, while exhibiting larger power gains due to the weaker hidden bias and reduced noise in the external data.

\section{Case Study}\label{sec:case_study}
It has become widely debated whether lobectomy, the current standard-of-care, remains optimal for patients with early-stage non-small-cell lung cancer (NSCLC). Lobectomy involves the removal of an entire lung lobe after tumor detection. In contrast, limited (sublobar) resection, including segmentectomy and wedge resection, is a less invasive alternative that removes only part of the lobe. While lobectomy is associated with a lower risk of local recurrence, it also results in greater loss of lung capacity and higher surgical morbidity, such as chronic pain and long-term respiratory impairment \citep{Khan2024}. Limited resection is therefore often preferred for patients with poor pulmonary function or very small tumors. For NSCLC patients with small ($\leq 2$ cm) peripheral tumors, recent randomized trials have shown that sublobar resection is noninferior to lobectomy in terms of disease-free survival \citep{saji2022segmentectomy, Altorki2023, Altorki2024}. These findings raise an important clinical concern regarding potential overtreatment: for patients who do not obtain meaningful benefit from lobectomy, the standard-of-care treatment may unnecessarily remove healthy lung tissue and impair postoperative function, while offering little survival advantage over limited resection. 

Consequently, this clinical context naturally leads to an individualized decision problem: which patients can benefit from limited resection instead of the standard lobectomy? Addressing this question requires moving beyond population-level comparisons toward reliable patient-level treatment selection.  However, such individualized decisions are inherently challenging due to uncertainty in treatment effect estimation and the risk of incorrectly recommending a non-standard treatment. To address this challenge, the proposed conformal inference framework that identifies patients who are likely to benefit from limited resection while explicitly accounting for uncertainty could be very helpful. Following the proposed framework, we apply the BH procedure to control the FDR among selected candidates. Among patients identified as benefiting more from limited resection, the proportion who would in fact benefit more from standard-of-care is controlled at a pre-specified level. In this case study, we set \(q=0.20\), as the goal is to screen for patients who may benefit from the less invasive surgical option rather than to make definitive treatment recommendations. This provides a principled and conservative mechanism for treatment-benefit screening: limited resection is recommended only for patients with sufficiently strong evidence of benefit, while candidates without such evidence continue to receive the standard-of-care. In this way, the proposed framework illustrates how uncertainty-aware screening help reduce potential overtreatment while maintaining caution in individualized treatment-benefit interpretation.

CALGB 140503 is a multi-center Phase III non-inferiority trial that enrolled patients with peripheral, node-negative NSCLC tumors smaller than $2$ cm and randomized them to receive either lobectomy or limited/sublobar resection \citep{Altorki2023}. The primary analysis yielded a hazard ratio for disease-free survival of 1.01 with a $90\%$ confidence interval of $(0.83, 1.24)$ and a hazard ratio for overall survival of 0.95 with a $90\%$ confidence interval of $(0.72, 1.26)$, which supports the non-inferiority of sublobar resection relative to the lobectomy procedure. 

Although the randomized design provides high internal validity, the sample size of the trial is limited for reliably learning heterogeneous treatment effects, as subgroup and individual-level analyses are often underpowered after conditioning on covariates \citep{kent_predictive_2020}. Therefore, the incorporation of external data to enhance model training and improve covariate coverage would be helpful. To address this limitation, we incorporate external data from the National Cancer Database (NCDB) to augment model training \citep{Bilimoria2008NCDB}. In this case study, NCDB is used to help train the treatment-effect model, whereas the conformal calibration step remains anchored in the randomized CALGB 140503 data. Thus, the registry contributes information for learning potential benefit patterns, while the RCT calibration set determines how the evidence for selecting beneficiaries is evaluated. 

The NCDB is a large, hospital-based oncology registry that captures approximately $70\%$ of newly diagnosed cancer cases in the United States. The NCDB records lobectomy and limited resection as treatment groups and shares a rich set of baseline covariates with CALGB 140503, which enables the integration of randomized and real-world data for outcome modeling. Compared to the randomized trial, the NCDB provides substantially greater sample size and broader covariate support, which are essential for training flexible models to estimate individualized treatment effects. More importantly, recent analyses of NCDB data indicate that the use of sublobar resection, particularly segmentectomy, has increased over time among low-risk patients, with five-year overall survival shown to be non-inferior to lobectomy \citep{Deng2025Temporal}. These findings are broadly consistent with the primary results of CALGB 140503, which therefore encourages the use of RWD to augment randomized evidence in individual treatment effect modeling. The goal of this case study is to apply the proposed conformal inference framework to identify patients who are more likely to benefit from limited resection within a candidate population before receiving any treatment, thereby supporting precision surgical decision-making while limiting unnecessary overtreatment.

The CALGB 140503 trial contains 697 observations. After removing observations with missing tumor size or overall survival time equal to zero, 694 participants remain, including 355 patients who underwent lobectomy and 339 who received limited resection. Following the eligibility criteria described in \citep{Altorki2023}, we further refine the NCDB data to improve comparability with the trial population by excluding observations that fail to meet the trial eligibility criteria or whose baseline covariates fall outside the range observed in CALGB 140503. After this filtering process, 14,742 observations from NCDB remain.

\begin{figure}[ht]
    \centering
    \includegraphics[width=\linewidth]{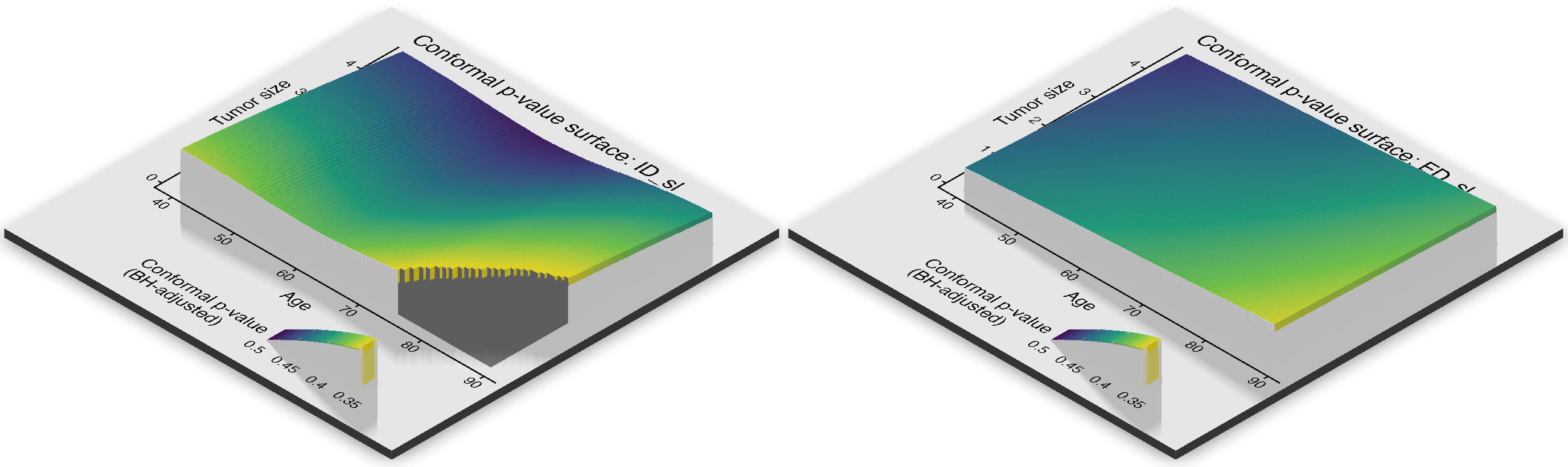}
    \caption{Comparison of conformal $p$-value surfaces under ID and ED when limited resection is treatment of interest and Super Learner is base learner}
    \label{fig:heatmap_sl_0.35}
\end{figure}

We use Restricted Mean Survival Time (RMST) as the outcome of interest. RMST is defined as $Y = \min(T, t^*)$, where $T$ denotes survival time and $t^*$ is a pre-specified truncation time. RMST represents the expected survival time restricted to $t^*$ and is commonly used as an alternative summary measure to the hazard ratio in survival analysis \citep{Uno2014}. In this case study, the CATE is defined as the difference in 5-year RMST between the two surgical arms given the covariate pattern $X=x$, i.e., $\tau(X=x)=\mathbb{E}\{Y(1)-Y(0)\mid X=x\}$ with $Y(a)=\min(T(a), 5)$ and $T(a)$ is the potential survival time under treatment $A=a$. To address right censoring in overall survival while maintaining compatibility with the potential outcomes framework, we transform the right-censored survival times into pseudo-RMST outcomes \citep{Andersen2010}. Because censoring may depend on baseline characteristics, we first examine the associations between baseline covariates and the censoring indicator. Based on this assessment, we estimate the censoring distribution using sex-stratified Kaplan-Meier curves within each data source. Within each data source and sex stratum, censoring is assumed to be independent of the underlying survival time. Pseudo-RMST values are then generated using the jackknife method implemented in the \texttt{eventglm} R package \citep{Sachs2022}. This procedure converts the censored survival outcomes in both CALGB 140503 and NCDB into 5-year pseudo-RMST values, which are treated as continuous outcomes in the subsequent analysis. Supplementary Material~\ref{app:empirical_stoch_caseStudy} further provides an empirical assessment of the stochastic-ordering condition in a case-study-like setting with censoring and pseudo-RMST construction.

The candidate population consists of patients who satisfy the eligibility criteria of CALGB 140503 before receiving any treatment, for whom only baseline covariates are observed. Such a population naturally arises in practice when identifying eligible patients for a follow-up study, such as an adaptive trial or a new trial targeting a specific subgroup. In this setting, the goal is to make treatment recommendations based solely on pre-treatment information. For illustration, we construct this candidate pool from NCDB. Specifically, after applying the CALGB 140503 eligibility criteria and restricting to the covariate range observed in the trial, we use nearest-neighbor matching on baseline covariates to identify NCDB patients most comparable to the CALGB 140503 trial population, and then randomly sample 800 patients from the matched NCDB population. After eligibility restriction and matching, the candidate group is considered as approximately exchangeable with the RCT calibration set.

Within the RCT data, we split the dataset into $40\%$ for model training, $40\%$ for the calibration set, and the remaining $20\%$ for nuisance model estimation required for constructing pseudo-individual treatment effects using the doubly-robust approach aforementioned. To improve the estimation of the CATE model, we randomly sample 3,000 observations from NCDB and combine them with the RCT training set. We implement the most robust and efficient base learners, which are BART and Super Learner, in this case study. Both training with and without the inclusion of NCDB data are investigated. Similar to the simulation study, we employ 5-fold cross-fitting to make efficient use of the available data. The baseline covariates included in the models are age, sex, race (white vs.\ non-white), tumor size, and histology (adenocarcinoma, squamous cell carcinoma, or other). For each random split and each base learner, we apply the BH procedure to the conformal \(p\)-values at target FDR level \(q=0.20\), yielding a split-specific set of selected candidates. Because the selected set may vary with the random split, we repeat the analysis 200 times and record the selection decision for each candidate under each base learner. These repeated-split results are summarized descriptively using selection frequencies to assess the stability of the identified beneficiary profiles. The formal FDR guarantee applies to the BH-selected set within each individual split, rather than to the aggregated selection frequencies.

\subsection{Case Study Results}
\begin{figure}[ht]
    \centering
    \includegraphics[width=\linewidth]{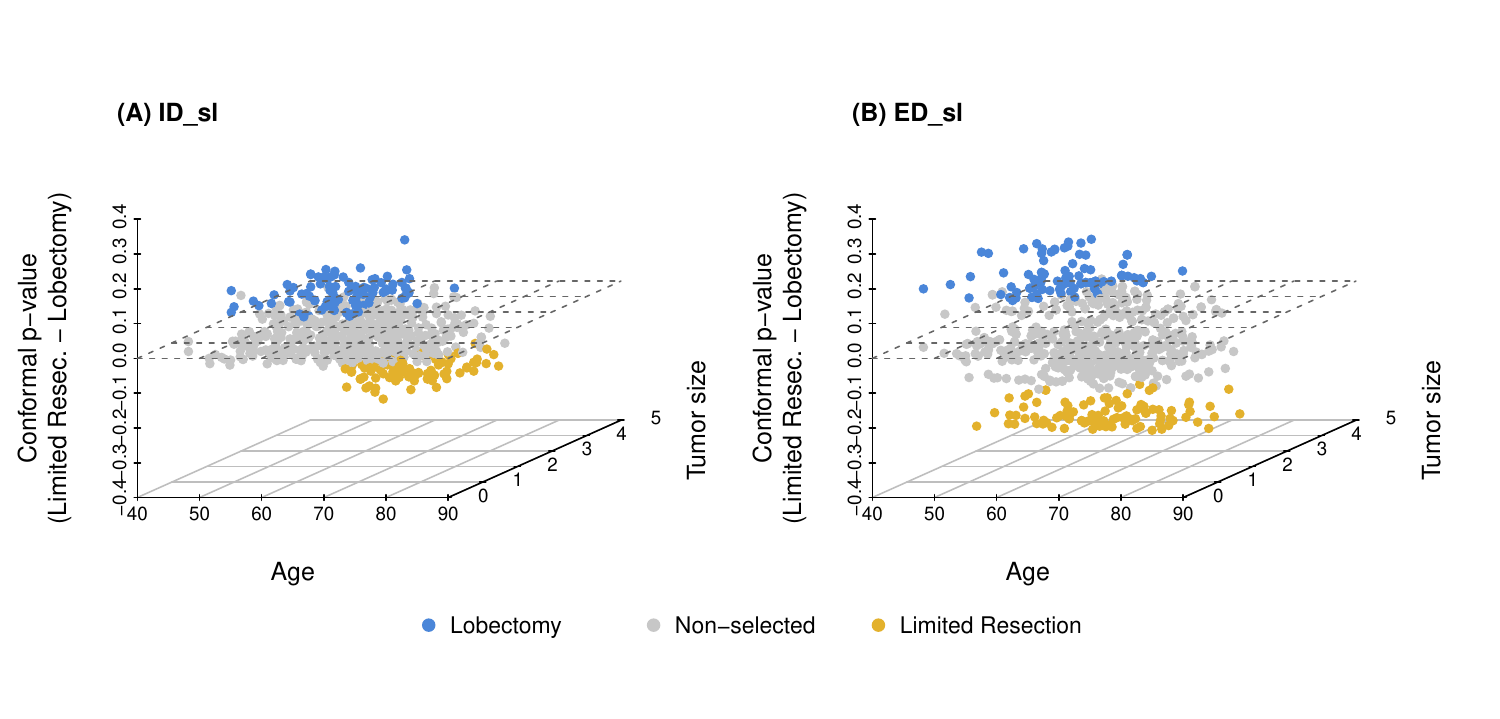}
    \caption{Distribution of conformal $p$-values regarding age and tumor sizes when using Super Learner as the base learner}
    \label{fig:sep_sl}
\end{figure}

Here, we present the case study results using Super Learner as the base learner, with corresponding BART results provided in Supplementary Material~\ref{app:case_results}. Figure \ref{fig:heatmap_sl_0.35} displays the BH-adjusted conformal $p$-value surfaces for testing $H_0:\tau(x)\le 0$, where $\tau(x)$ is defined as the CATE of limited resection relative to lobectomy in 5-year RMST. A smaller conformal $p$-value reflects stronger evidence that the covariate profile (or the patient subgroup) is associated with conditional benefit from limited resection. The gray regions, where BH-adjusted conformal $p$-values fall below $0.35$, are used only to highlight areas with relatively stronger evidence in the surface plot and are not used as the FDR-controlled selection rule. When training on RCT data alone (ID), the surface is notably irregular. Specifically, the signal for limited-resection benefit is present among younger patients, weakens in the middle age range, and reappears among older patients. Given the limited sample size of CALGB 140503 for estimating heterogeneous treatment effects, this pattern may reflect estimation instability rather than a clinically meaningful structure. Incorporating external NCDB data during model training (ED) yields a smoother surface, with stronger evidence for limited-resection benefit appearing among older patients and those with smaller tumors. This pattern is more stable and clinically plausible, which is aligned with what prior subgroup analyses have reported \citep{Altorki2023,Mao2025pub}. From this, the smoother ED surface may reflect improved covariate coverage and reduced estimation variability, which together may help stabilize the learned treatment effect surface and improve the interpretability of the results.

For descriptive visualization, we also examine candidate-level evidence in both treatment directions, corresponding to potential benefit from lobectomy and potential benefit from limited resection. Thereby, the candidate population is able to be summarized into three groups: candidates with stronger evidence of benefit from lobectomy, candidates with stronger evidence of benefit from limited resection, and candidates without sufficient evidence favoring either direction. For each treatment direction, directional conformal $p$-values are computed and the BH procedure is applied at an FDR level of $20\%$ to define the error-controlled selected set. Candidates selected in both directions are not assigned to either group and are instead classified as non-selected, because the evidence does not clearly favor one treatment direction. In addition, Figure~\ref{fig:sep_sl} displays the 100 candidates with the smallest $p$-values in each treatment direction.  This fixed-size display is descriptive and is intended only to visualize the separation between candidate profiles with stronger evidence for each treatment direction. In the figure, blue, yellow, and gray points correspond to the lobectomy-evidence, limited-resection-evidence, and non-selected groups, respectively.

Compared with training on RCT data alone (ID), incorporating external NCDB data (ED) produces a clearer separation between candidates with stronger evidence for the two treatment directions. In the ID analysis, candidates highlighted for lobectomy and limited resection are more concentrated around zero in the conformal $p$-value contrast, indicating weaker and less distinguishable signals. In the ED analysis, the highlighted candidates present a more noticeable separation, while non-highlighted candidates remain centered around zero. Therefore, leveraging external data enhances the ability to distinguish meaningful treatment effect signals from noise, leading to more stable and interpretable individualized treatment recommendations.

\section{Conclusion and Discussion}\label{sec:discussion}

To provide an uncertainty-aware individual-level selection rule before observing candidate outcomes, we propose a model-agnostic conformal-based multiple-testing framework for beneficiary identification based on estimated individual treatment effects. The primary goal is to identify a subgroup of individuals who are likely to benefit from the treatment of interest, and thus can be useful for assisting treatment selection in practice or guiding future trial enrollment. Unlike approaches that directly threshold estimated ITEs, our framework turns this estimation-based selection problem into a testing problem. Because the conformal layer is separated from the choice of treatment effect learner, the framework can be integrated with a broad class of statistical and machine learning models. By applying the Benjamini-Hochberg procedure to conformal $p$-values, we explicitly account for the multiplicity issue that naturally arises when selecting a subgroup from a larger candidate population. In this way, our proposed method translates pseudo-ITE estimates into valid selection decisions with uncertainty quantification and FDR control.

When the RCT sample size is limited, we further consider borrowing external data, such as RWD, to strengthen the model-training step. In our framework, external data are used to improve treatment effect estimation, while RCT-based conformal calibration provides the reference distribution for uncertainty-aware selection. As shown in the simulation studies, when the external data are informative and reasonably comparable to the RCT population, incorporating them can improve the prediction model and help select more true beneficiaries. In the case study, we further examine which baseline characteristic patterns are more likely to be associated with benefit from limited resection or lobectomy. The case study illustrates how the framework can use trial and external real-world data to identify candidates with evidence of benefit from limited resection.

This framework is particularly relevant for enrichment designs and precision medicine applications. For example, clinicians may hope to use a data-driven tool to support treatment decision making for individual patients, while trial investigators may wish to identify a subset of candidates for second-stage enrollment or for a follow-up study. In these settings, simply ranking or thresholding estimated ITEs may ignore the uncertainty in individual-level treatment effect estimation. The proposed conformal testing framework provides a more cautious alternative by selecting beneficiaries only when there is sufficient evidence that their treatment effects exceed clinically meaningful thresholds.

The proposed framework is most directly suited to settings where candidate subjects come from the same target population as the initial RCT, such as follow-up target trials, second-stage enrollment, or adaptive enrichment designs after an initial RCT. In these cases, exchangeability between the calibration and candidate samples is a reasonable working assumption. In practice, investigators may also wish to apply an RCT-derived screening rule to a different target population, which is closely related to the generalizability and transportability problem \citep{Stuart2015,Dahabreh2019,Lee2023,Liu2024,liu2025estimandfocus}. Although this setting is beyond the main scope of the present paper, one possible extension is weighted conformal inference, where calibration scores are reweighted toward the target population if density ratios, analogous to sampling scores in causal generalizability, can be estimated \citep{tibshirani2019conformal}. We provide pseudo-code for this possible extension in Algorithm~\ref{alg:weighted_conformal_screening} in Supplementary Material \ref{app:weighted_conformal}. When the candidate population differs from the RCT population in ways not fully captured by covariate-shift weighting, such as temporal changes or distribution drift, additional robust calibration tools may also be needed \citep{barber2023conformal}.

A common concern in RWD-assisted learning is that external data may differ systematically from the RCT population because of differences in patient mix, treatment assignment, outcome measurement, or temporal effects. These differences may affect the nuisance models or treatment effect models learned during training. In the current framework, however, external data are used only for model training, and conformal calibration remains anchored in the RCT data. Therefore, potential external bias mainly affects efficiency and power, rather than directly changing the calibration reference distribution for uncertainty-aware selection. We do not use a test-then-pool \citep{yang2023elastic} or selective borrowing method \citep{Gao2024,Zhu2024,liu2025hybridControl}, because using the same RCT data both to decide borrowing and to perform beneficiary selection may introduce a double-dipping concern and make downstream error control harder to interpret. Instead, we take a more transparent approach by trimming the external data according to the published RCT eligibility criteria before model training, which is more realistic  when investigators want to put the framework into practice. More adaptive borrowing strategies, such as including a data-source indicator in flexible learners to adjust borrowing depending on external data similarity \citep{zhou2021incorporating}, may further improve efficiency, but they can remain implicitly data-driven and model-dependent. How to balance such efficiency gains with valid downstream selection remains an important direction for future work.

In this work, FDR control is achieved by applying the BH procedure to a fixed batch of candidate subjects, which is natural for enrichment studies or follow-up trials where the candidate pool is available before selection. In some clinical settings, candidates may arrive sequentially, and treatment or enrollment decisions may need to be made in real time. In such settings, each candidate's conformal $p$-value can still be computed in real time because it depends only on the fixed RCT calibration set and not on other candidates. However, the BH step itself is a batch procedure. Real-time treatment or enrollment decisions would therefore require replacing BH with an online FDR procedure, such as alpha-investing-based rules including LORD, LOND, SAFFRON, and ADDIS \citep{javanmard2015, chen2017sequentialmultipletesting, Ramdas2017, javanmard2018online, Ramdas2018, tian2019addis}. Extending the proposed framework to online FDR control is therefore a useful future direction.
Finally, the validity of conformal $p$-values is marginal, rather than conditional on each covariate value. In other words, the guarantee holds over the joint distribution of calibration and candidate subjects, but it does not ensure exact validity for every fixed $X=x$. This is expected because the conformal $p$-value is calibrated using the available calibration sample as a whole. If subgroup-specific guarantees are of interest, one can calibrate within pre-specified subgroups, such as sex, histology, or other categorical clinical variables. For continuous covariates, clinically meaningful bins can be defined before calibration so that subgroup-wise calibration can still be applied. However, for arbitrary continuous covariate values, exact finite-sample conditional validity is generally difficult to maintain without additional assumptions \citep{angelopoulos2026theoreticalfoundations}.

\section{Funding}
This project is supported by the Food and Drug Administration (FDA) of the U.S. Department of Health and Human Services (HHS) as part of a financial assistance award U01FD007934 funded by FDA/HHS. It is also supported by the National Institute on Aging of the National Institutes of Health under Award Number R01AG06688. The contents are those of the authors and do not necessarily represent the official views of, nor an endorsement by, the FDA/HHS, the National Institutes of Health, or the U.S. Government.

\section{Conflicts of Interest}
No conflict of interest.

\section{Data Availability}
The complete replication codes and simulated data are available at \url{https://github.com/jiajunliu0511/ConfSelection}. The NCDB data may be made available upon reasonable request, whereas the CALGB 140503 data cannot be shared due to restricted access.

\bibliographystyle{plainnat}
\bibliography{ref}

\appendix

\setcounter{equation}{0}
\renewcommand{\theequation}{S\arabic{equation}}
\setcounter{table}{0}
\renewcommand{\thetable}{S\arabic{table}}
\setcounter{figure}{0}
\renewcommand{\thefigure}{S\arabic{figure}}
\setcounter{theorem}{0}
\renewcommand{\thetheorem}{S\arabic{theorem}}
\setcounter{lemma}{0}
\renewcommand{\thelemma}{S\arabic{lemma}}
\setcounter{remark}{0}
\renewcommand{\theremark}{S\arabic{remark}}

\newpage
\begin{center}
  {\large\bf SUPPLEMENTARY MATERIAL}
\end{center}

\newcounter{rmk}
\newcommand\rmk[1]{\vspace*{1mm} \par \stepcounter{rmk}{\noindent \bf Remark \thermk}. {#1}\vspace*{1mm}}

\section{From ITE coverage to valid conformal \texorpdfstring{$p$}{p}-values}\label{app:coverage_to_pvalue}
In this section, we show how a valid one-sided conformal prediction set for the true individual treatment effect (ITE) induces a valid conformal $p$-value. This connection is needed because the conformal meta-learner result of \citet{alaa2024conformal} is stated in terms of prediction interval coverage, whereas our selection procedure is based on conformal $p$-values.

Following the notation in the main text, for candidate individual $j$, we write $\tau_j = \tau(X_{n+m+j})$ for the true treatment-benefit function evaluated at the candidate's covariates. We consider the one-sided null hypothesis $H_{0j}:\tau_j \le c$, where $c$ is a pre-specified clinically meaningful threshold. For simplicity, we take
$c_j\equiv c$ throughout this section.

Suppose that, under the stochastic-ordering conditions of \citet{alaa2024conformal}, the pseudo-outcome conformal procedure yields a one-sided prediction set for the true ITE of the form $\hat C_\alpha(X_j)= [L_\alpha(X_j),\infty),$ satisfying the marginal coverage guarantee 
\begin{equation*}
    \mathbb P\{\tau_j\in \hat C_\alpha(X_j)\}=\mathbb P\{\tau_j \ge L_\alpha(X_j)\} \ge 1-\alpha,
\end{equation*}
Equivalently,
\[
    \mathbb P\{\tau_j < L_\alpha(X_j)\}\le \alpha.
\]
For each $\alpha\in[0,1]$, this lower conformal bound induces a one-sided level-$\alpha$ test of $H_{0j}$ that rejects whenever the entire conformal lower set lies above the null threshold, i.e., $c < L_\alpha(X_j)$.

Assume that the resulting rejection rule is monotone in $\alpha$, so that rejection at a smaller nominal level implies rejection at any larger nominal level. This monotonicity holds for the usual conformal quantile construction of $L_\alpha(X_j)$. By inverting this family of tests, we define the conformal $p$-value as
\begin{equation*}
    p_j=\inf\{\alpha\in[0,1]: c < L_\alpha(X_j)\}.
\end{equation*}
Thus, $p_j$ is the smallest nominal level at which the one-sided conformal test rejects the null hypothesis.

\begin{lemma}[Coverage implies generalized validity of the conformal $p$-value]
\label{lem:coverage_to_pvalue}
If the one-sided conformal set satisfies
\[
    \mathbb P\{\tau_j < L_\alpha(X_j)\}\le \alpha
\]
for every $\alpha\in[0,1]$, then the conformal $p$-value defined above satisfies
\[
    \mathbb P\{p_j\le \alpha,\ \tau_j\le c\}\le \alpha,
    \qquad \forall \alpha\in[0,1].
\]
\end{lemma}
\begin{proof}
Fix $\alpha\in[0,1]$, by the definition of the conformal $p$-value defined in \eqref{eq:pvalue_deterministic}, the event $\{p_j\le \alpha\}$ implies
\[
    c < L_\alpha(X_j).
\]
If the null hypothesis is true, i.e., $\tau_j\le c$, then
\[
    \tau_j\le c < L_\alpha(X_j).
\]
Therefore,
\[
    \{p_j\le \alpha,\ \tau_j\le c\}
    \subseteq
    \{\tau_j < L_\alpha(X_j)\}.
\]
Taking probabilities on both sides gives
\[
    \mathbb P\{p_j\le \alpha,\ \tau_j\le c\}
    \le
    \mathbb P\{\tau_j < L_\alpha(X_j)\}
    \le \alpha,
\]
where the last inequality follows from the one-sided conformal coverage
guarantee. 
\end{proof}

\section{False discovery rate control for pseudo-ITE conformal selection}\label{app:FDR_control}
In this section, we establish FDR control for the proposed selection procedure. The previous section shows that a valid one-sided conformal prediction set for the true ITE induces a generalized-valid conformal $p$-value for testing
$H_{0j}: \tau_j \le c_j,$ where $\tau_j=\tau(X_{n+m+j})$ and $c_j$ is a pre-specified clinical threshold. 

For notational clarity, we distinguish the following scores. Let $n_c=|\mathcal D_{\mathrm{Calib}}|$ denote the calibration sample size and let $n_t$ denote the number of candidate individuals. For candidate $j=1,\ldots,n_t$, define
\[
    \tau_j=\tau(X_{n+m+j}), \qquad
    \hat V_j = V(X_{n+m+j},c_j), \qquad
    V_j^* = V(X_{n+m+j},\tau_j).
\]
Here, $\hat V_j$ is the null-imputed boundary score used to construct the conformal $p$-value, while $V_j^*$ is the oracle candidate score that would be available if the true treatment benefit $\tau_j$ were observed.
For calibration subject $i\in\mathcal D_{\mathrm{Calib}}$, let
\[
    V_i = V(X_i,Y_i'), \qquad
    V_i^* = V(X_i,\tau(X_i)),
\]
where $V_i$ denotes the pseudo-outcome-based calibration score used in the implementable procedure and $V_i^*$ is the corresponding oracle calibration score. 

In the main text, we define the conformal $p$-value
\[
    p_j=\frac{\sum_{i=1}^{n_c}\mathbb{I}\{V_i<\hat V_j\}+ U_j\left(1+\sum_{i=1}^{n_c}\mathbb{I}\{V_i=\hat V_j\}\right)}{n_c+1},
\]
where $U_j\sim\mathrm{Unif}(0,1)$ is used for tie-breaking. For the finite-sample exchangeability proof, following the deterministic version of \citet{Jin2023SelectionbyConP}, we use

\begin{equation}
\label{eq:pvalue_deterministic}
    p_j^{\mathrm{d}}
    =
    \frac{1+\sum_{i=1}^{n_c}\mathbb I\{V_i<\hat V_j\}}{n_c+1}.
\end{equation}
Under the no-ties condition assumed in Theorem~\ref{thm:fdr}, this deterministic $p$-value is a conservative rank-based version that avoids auxiliary tie-breaking randomness. For notational simplicity, we write \(p_j\equiv p_j^{\mathrm{d}}\) throughout this section. We now show how these $p$-values can be adjusted through Benjamini--Hochberg (BH) procedure to control the false discovery rate (Theorem~\ref{thm:fdr})

\begin{proof}[Proof of Theorem~\ref{thm:fdr}]
Throughout the proof, we condition on the training stage, including nuisance estimation, pseudo-outcome construction, and the fitted treatment effect model. Thus, the conformal score function $V$ and the fitted functions entering the conformal scores are treated as fixed.
Following the notations aforementioned, $\hat V_j=V(X_{n+m+j},c_j)$ is the boundary score evaluated under null and $V_j^*=V(X_{n+m+j},\tau_j)$ is the oracle score evaluated at unknown $\tau_j=\tau(X_{n+m+j})$. For calibration subjects, let \(\{V_i:i=1,\ldots,n_c\}\) and \(i\in \mathcal{D}_{\mathrm{Calib}}\), denote the pseudo-outcome-based calibration conformal scores used in the conformal ranking step. 

First, when assume there is no tie among $\{V_i:i=1,\ldots,n_c\} \cup \{\hat V_{\ell}:\ell=1,\ldots,n_t\}$, we define the oracle $p$-value as 
    \begin{equation}\label{eq:pvalue_oracle}
        p_j^*=\frac{1+\sum_{i=1}^{n_c}\mathbb{I}\{V_i<{V}_j^*\}}{n_c+1},
    \end{equation}
where $n_c=|\mathcal D_{\mathrm{Calib}}|$.
For candidates $j=1,\ldots,n_t$ and $\ell\ne j$, we define a set of modified $p$-values 
    \begin{equation}\label{eq:pvalue_modified}
        p_\ell^{(j)}=\frac{\sum_{i=1}^{n_c}\mathbb{I}\{V_i<\hat{V}_\ell\}+\mathbb{I}\{{V}^*_{j}<\hat{V}_\ell\}}{n_c+1},
    \end{equation}
which is used for proof only and not for real practice. It is the conformal rank of the score $\hat{V}_\ell$ when the reference set is enlarged from $\{V_i:i\in\mathcal D_{\mathrm{Calib}}\}$ to $\{V_i:i\in\mathcal D_{\mathrm{Calib}}\}\cup\{V_j^*\}.$

For any vector of candidate $p$-values $(a_1,\ldots,a_{n_t})$, let $\mathcal R(a_1,\ldots,a_{n_t})\subseteq\{1,\ldots,n_t\}$ denote the rejection set returned by the BH procedure at level $q$ when applied to this vector. The original rejection set is
    \[
    \mathcal R=\mathcal R(p_1,\ldots,p_{n_t}).
    \]
We compare $\mathcal R$ with the modified rejection set
    \[
    \mathcal R_j^* = \mathcal R (p_1^{(j)},\ldots,p_{j-1}^{(j)},p_j^*,p_{j+1}^{(j)},\ldots,p_{n_t}^{(j)}).
    \]

Consider the event
\[
    \{\tau_j=\tau(X_{n+m+j})\le c_j,\ j\in\mathcal R\}.
\]
On this event, monotonicity of $V$ in its second argument gives
\[
    V_j^*=V(X_{n+m+j},\tau_j) \le V(X_{n+m+j},c_j)=\hat V_j.
\]
Therefore, based on $p_j^*$ and $p_j$ defined in \eqref{eq:pvalue_deterministic} and \eqref{eq:pvalue_oracle}, we have $p_j^*\le p_j$. 

Now, we consider any other candidates $\ell\ne j$. Using the modified $p$-value defined in \eqref{eq:pvalue_modified}, we consider two cases for any $\ell\ne j$. First, if $\hat V_\ell>\hat V_j$, then $V_j^*\le \hat V_j<\hat V_\ell$, so
\[
    p_\ell^{(j)}=
    \frac{1+\sum_{i=1}^{n_c}\mathbb{I}\{V_i<\hat V_\ell\}}{n_c+1}=p_\ell.
\]
Second, if $\hat V_\ell<\hat V_j$, then the deterministic conformal $p$-value is nondecreasing in the candidate boundary score, so
\[
    p_\ell=\frac{1+\sum_{i=1}^{n_c}\mathbb{I}\{V_i<\hat V_\ell\}}{n_c+1}
    \le\frac{1+\sum_{i=1}^{n_c}\mathbb{I}\{V_i<\hat V_j\}}{n_c+1} =p_j.
\]
Since $j\in\mathcal R$ and $p_\ell\le p_j$, the BH step-up rule implies
$\ell\in\mathcal R$. Therefore,
\begin{align*}
    p_\ell^{(j)}
    &= \frac{\sum_{i=1}^{n_c}\mathbb{I}\{V_i<\hat{V}_\ell\}+\mathbb{I}\{{V}^*_{j}<\hat{V}_\ell\}}{n_c+1}\\
    &\le \frac{\sum_{i=1}^{n_c}\mathbb{I}\{V_i<\hat{V}_\ell\}+1}{n_c+1} \qquad \text{(by indicator term is at most 1)}\\
    &= p_\ell
    \le p_j =\frac{\sum_{i=1}^{n_c}\mathbb{I}\{V_i<\hat{V}_j\}+1}{n_c+1}
\end{align*} 
For Case 1, $\hat V_\ell>\hat V_j$ implies $p_\ell\ge p_j$ and $p_\ell^{(j)}=p_\ell$. Thus, all $p$-values for other candidates $\ell\ne j$ at or above $p_j$ remain unchanged after the replacement. For Case 2, $\hat V_\ell<\hat V_j$ implies $p_\ell\le p_j$ and $p_\ell^{(j)}\le p_\ell\le p_j$. Thus, all modified $p$-values corresponding to scores below $\hat V_j$ remain no larger than $p_j$. Combining the two cases, replacing $p_\ell$ by $p_\ell^{(j)}$ for all $\ell\ne j$ does not change the ordering of the $p$-values around the reference value $p_j$. All $p$-values larger than $p_j$ remain unchanged, while all $p$-values no larger than $p_j$ remain no larger than $p_j$ after replacement. In addition, on the event $\{\tau_j\le c_j,\ j\in\mathcal R\}$, monotonicity gives $p_j^*\le p_j$. Since $j$ is already rejected, decreasing its $p$-value from $p_j$ to $p_j^*$ cannot remove any rejection under the BH step-up rule. Therefore, by the step-up nature of the BH procedure, this replacement does not
change the rejection set on the event $\{\tau_j\le c_j,\ j\in\mathcal R\}$, which means
\begin{align}\label{eq:rejection_set_equal}
    \mathcal R
    &=\mathcal R(p_1^{(j)},\ldots,p_{j-1}^{(j)},p_j,p_{j+1}^{(j)},\ldots,p_{n_t}^{(j)})\notag\\ 
    &=\mathcal R(p_1^{(j)},\ldots,p_{j-1}^{(j)},p_j^*,p_{j+1}^{(j)},\ldots,p_{n_t}^{(j)})=:\mathcal R_j^*
\end{align}
    
Then, let $R_j=\mathbb{I}\{j\in \mathcal{R}\}$ as the rejection indicator for candidates $j=1,\ldots,n_t$, and $|\mathcal{R}|=\sum_{j=1}^{n_t}R_j$. As there must exist a $r\in [0,n_t]$ such that $|\mathcal{R}|=r$, so $\sum_{r=0}^{n_t}\mathbb{I}\{|\mathcal{R}|=r\}=1$. We start from the definition of FDR
\begin{align*}
    \mathrm{FDR} 
    &=\mathbb E\left[\frac{\sum_{j=1}^{n_t}\mathbb {I}\{\tau_j\le c_j\}R_j}{1\vee \sum_{j=1}^{n_t}R_j}\right]  \\
    &=\mathbb E\left[\sum_{r=0}^{n_t} \frac{\sum_{j=1}^{n_t}\mathbb {I}\{\tau_j\le c_j\}R_j}{1\vee \sum_{j=1}^{n_t}R_j}\mathbb{I}\{|\mathcal{R}|=r\}\right]
\end{align*}
If $r=0$, then $\sum_{j=1}^{n_t}\mathbb {I}\{\tau_j\le c_j\}R_j=0$ as $R_j\equiv0$ for all $j$. If $r\ge 1$, then $1\vee \sum_{j=1}^{n_t}R_j=r$. Therefore, we can rewrite FDR as
\begin{align*}
    \mathrm{FDR} 
    &=\mathbb E\left[\sum_{r=1}^{n_t} \frac{\sum_{j=1}^{n_t}\mathbb {I}\{\tau_j\le c_j\}R_j}{1\vee \sum_{j=1}^{n_t}R_j}\mathbb{I}\{|\mathcal{R}|=r\}\right]\\
    &=\sum_{j=1}^{n_t}\sum_{r=1}^{n_t}\frac{1}{r}\mathbb E\left[\mathbb {I}\{\tau_j\le c_j\}R_j\mathbb{I}\{|\mathcal{R}|=r\}\right]\\
    &=\sum_{j=1}^{n_t}\sum_{r=1}^{n_t}\frac{1}{r}\mathbb E\left[\mathbb {I}\{\tau_j\le c_j\}\mathbb{I}\{|\mathcal{R}|=r\}\mathbb{I}\{p_j\le qr/n_t\}\right] \\
    &\qquad\qquad\qquad\qquad\text{(by BH procedure, $R_j=1\Leftrightarrow p_j\le q|\mathcal{R}|/n_t$)}\\
    &\leq \sum_{j=1}^{n_t}\sum_{r=1}^{n_t}\frac{1}{r}\mathbb E\left[\mathbb {I}\{\tau_j\le c_j\}\mathbb{I}\{|\mathcal{R}_j^*|=r\}\mathbb{I}\{p_j^*\le qr/n_t\}\right]\\
    & \qquad\qquad\qquad\qquad\text{(by $\mathcal R=\mathcal R_j^*$ in \eqref{eq:rejection_set_equal} and $p_j^*\le p_j$ shown above on \(\{\tau_j\le c_j,\ j\in\mathcal R\}\))}\\
    & \leq \sum_{j=1}^{n_t}\sum_{r=1}^{n_t}\frac{1}{r}\mathbb E\left[\mathbb{I}\{|\mathcal{R}_j^*|=r\}\mathbb{I}\{p_j^*\le qr/n_t\}\right]
\end{align*}
Next, by the BH procedure applied to the modified rejection set \(\mathcal R_j^*\), on the event \(\{|\mathcal R_j^*|=r\}\),
\begin{align*}
     p_j^*\le qr/n_t\quad\Longleftrightarrow\quad j\in\mathcal R_j^*,
\end{align*}
and thus
\begin{align*}
    \mathrm{FDR} \leq \sum_{j=1}^{n_t}\sum_{r=1}^{n_t}\frac{1}{r}\mathbb E\left[\mathbb{I}\{|\mathcal{R}_j^*|=r\}\mathbb{I}\{j\in \mathcal{R}_j^*\}\right].
\end{align*}
Finally, because replacing the already rejected \(p\)-value \(p_j^*\) by zero does not change the rejection set, therefore, on the event \(\{j\in\mathcal R_j^*\}\), 
\begin{align*}
    \mathcal R_j^*
    &=\mathcal R(p_1^{(j)},\ldots,p_{j-1}^{(j)},p_j^*,p_{j+1}^{(j)},\ldots,p_{n_t}^{(j)})\\
    &=\mathcal R(p_1^{(j)},\ldots,p_{j-1}^{(j)},0,p_{j+1}^{(j)},\ldots,p_{n_t}^{(j)})=:\mathcal R_{j\to 0}^*.
\end{align*}
Thus,
\begin{align*}
    \mathrm{FDR}
&\le \sum_{j=1}^{n_t}\sum_{r=1}^{n_t} \frac{1}{r} \mathbb E \left[\mathbb{I}\{|\mathcal R_{j\to0}^*|=r\}\mathbb{I}\{p_j^*\le q|\mathcal R_{j\to0}^*|/n_t\}\right] \\
&=\sum_{j=1}^{n_t}\mathbb E\left[\frac{\mathbb{I}\{p_j^*\le q|\mathcal R_{j\to0}^*|/n_t\}}{1\vee |\mathcal R_{j\to0}^*|}\right].
\end{align*}

Next, in order to show $\mathrm{FDR}\le q$, we only need to show 
\begin{equation}\label{eq:med_ineq_proof}
    \mathbb E\left[\frac{\mathbb{I}\{p_j^*\le q|\mathcal R_{j\to0}^*|/n_t\}}{1\vee |\mathcal R_{j\to0}^*|}\right]\le \frac{q}{n_t}.
\end{equation}

As we set the $p_j^*$ to zero, the rejection set $\mathcal R_{j\to0}^*$ no longer depends on $p_j^*$. For fixed $j\in\{1,\ldots,n_t\}$, define the unordered set
\[
    \mathcal U_j=[V_1,\ldots,V_{n_c},V_j^*],
\]
where $V_j^*=V(X_{n+m+j},\tau_j)$ is the oracle score for candidate $j$. By construction in \eqref{eq:pvalue_modified}, for each $\ell\ne j$, $p_\ell^{(j)}$ depends on $\{V_1,\ldots,V_{n_c},V_j^*\}$ only through their unordered values. Therefore, $\{p_\ell^{(j)}:\ell\ne j\}$ is invariant to any permutation of $ \{V_1,\ldots,V_{n_c},V_j^*\}$.

By Assumption~\ref{as:score_exchangeability}, the scores $\{V_1,\ldots,V_{n_c},V_j^*\}$ are exchangeable conditional on the boundary scores $\{\hat V_\ell:\ell\ne j\}.$ Thus, conditional on the unordered multiset $\mathcal U_j$ and the boundary scores $\{\hat V_\ell:\ell\ne j\}$, the rank of $V_j^*$ among $\{V_1,\ldots,V_{n_c},V_j^*\}$ is uniformly distributed. Equivalently, the oracle conformal $p$-value $p_j^*$ is conditionally super-uniform. This means that, for any random variable $t$ measurable with respect to $\mathcal U_j$ and $\{\hat V_\ell:\ell\ne j\}$,
\begin{equation}\label{eq:exch}
    \mathbb P\left(p_j^*\le t\,\middle|\,\mathcal U_j,\{\hat V_\ell:\ell\ne j\}\right)\le t.
\end{equation}

Next, because $\mathcal R_{j\to 0}^*$ is computed from $(p_1^{(j)},\ldots,p_{j-1}^{(j)},0,p_{j+1}^{(j)},\ldots,p_{n_t}^{(j)})$, $\mathcal R_{j\to 0}^*$ depends only on $\{p_\ell^{(j)}:\ell\ne j\},$. By construction, each $p_\ell^{(j)}$ is a function of $\mathcal U_j$ and $\{\hat V_\ell:\ell\ne j\}$. Therefore, $\mathcal R_{j\to 0}^*$ depends only on $\mathcal U_j$ and $\{\hat V_\ell:\ell\ne j\}$, but not on $p_j^*$. Then, $|\mathcal R_{j\to 0}^*|$ is measurable with respect to
$\mathcal U_j$ and $\{\hat V_\ell:\ell\ne j\}$. 

By the tower property,
\begin{align*}
\mathbb E\left[\frac{\mathbb I\{p_j^*\le q|\mathcal R_{j\to 0}^*|/n_t\}}{1\vee |\mathcal R_{j\to 0}^*|}\right] &=\mathbb E\left[\mathbb E\left\{\frac{\mathbb I\{p_j^*\le q|\mathcal R_{j\to 0}^*|/n_t\}}{1\vee |\mathcal R_{j\to 0}^*|}\,\middle|\,\mathcal U_j,\{\hat V_\ell:\ell\ne j\}\right\}\right]\\
&=\mathbb E\left[\frac{\mathbb P\left(p_j^*\le q|\mathcal R_{j\to 0}^*|/n_t\,\middle|\,\mathcal U_j,\{\hat V_\ell:\ell\ne j\}\right)}{1\vee |\mathcal R_{j\to 0}^*|}\right] \qquad\text{(by $|\mathcal R_{j\to 0}^*|$ is fixed)}
\end{align*}

Then, let $t=q|\mathcal R_{j\to 0}^*|/n_t$, which is measurable with respect to $\mathcal U_j$ and $\{\hat V_\ell:\ell\ne j\}$. Therefore, by \eqref{eq:exch},
\begin{align*}
\mathbb E\left[\frac{\mathbb I\{p_j^*\le q|\mathcal R_{j\to 0}^*|/n_t\}}{1\vee |\mathcal R_{j\to 0}^*|}\right]
&=\mathbb E\left[\frac{\mathbb P\left(p_j^*\le q|\mathcal R_{j\to 0}^*|/n_t\,\middle|\,\mathcal U_j,\{\hat V_\ell:\ell\ne j\}\right)}{1\vee |\mathcal R_{j\to 0}^*|}\right]\\
&\le \mathbb E\left[ \frac{q|\mathcal R_{j\to 0}^*|/n_t}{1\vee |\mathcal R_{j\to 0}^*|}\right]\le\mathbb E\left[ \frac{q}{n_t}\right]=\frac{q}{n_t}.
\end{align*}
For now, we establish \eqref{eq:med_ineq_proof}. Summing this bound over $j=1,\ldots,n_t$ yields
\begin{align*}
    \mathrm{FDR}\le \sum_{j=1}^{n_t}\frac{q}{n_t} =q,
\end{align*}
which completes the proof.

\end{proof}

\section{Additional simulation results}\label{app:sim_results}
In this section, we provide additional simulation results with the hidden bias magnitude $b=0.4$ in Figure \ref{fig:FDR_all_Yshift_0.4} and \ref{fig:Power_all_Yshift_0.4}.

\begin{figure}[!ht]
    \centering
    \includegraphics[width=\linewidth]{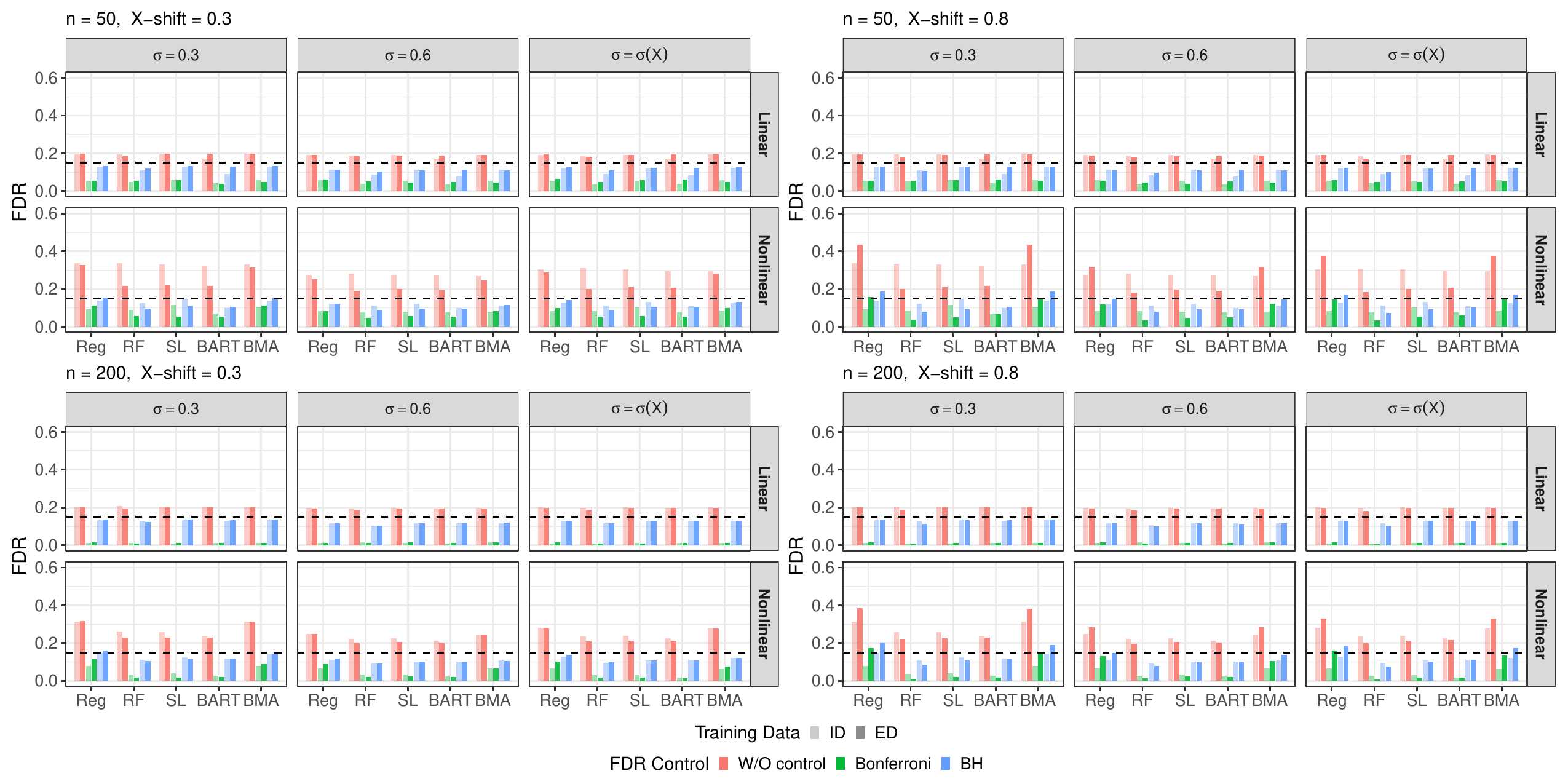}
    \caption{FDR across different scenarios given outcome drift ($b=0.4$)}
    \label{fig:FDR_all_Yshift_0.4}
\end{figure}
\begin{figure}[!ht]
    \centering
    \includegraphics[width=\linewidth]{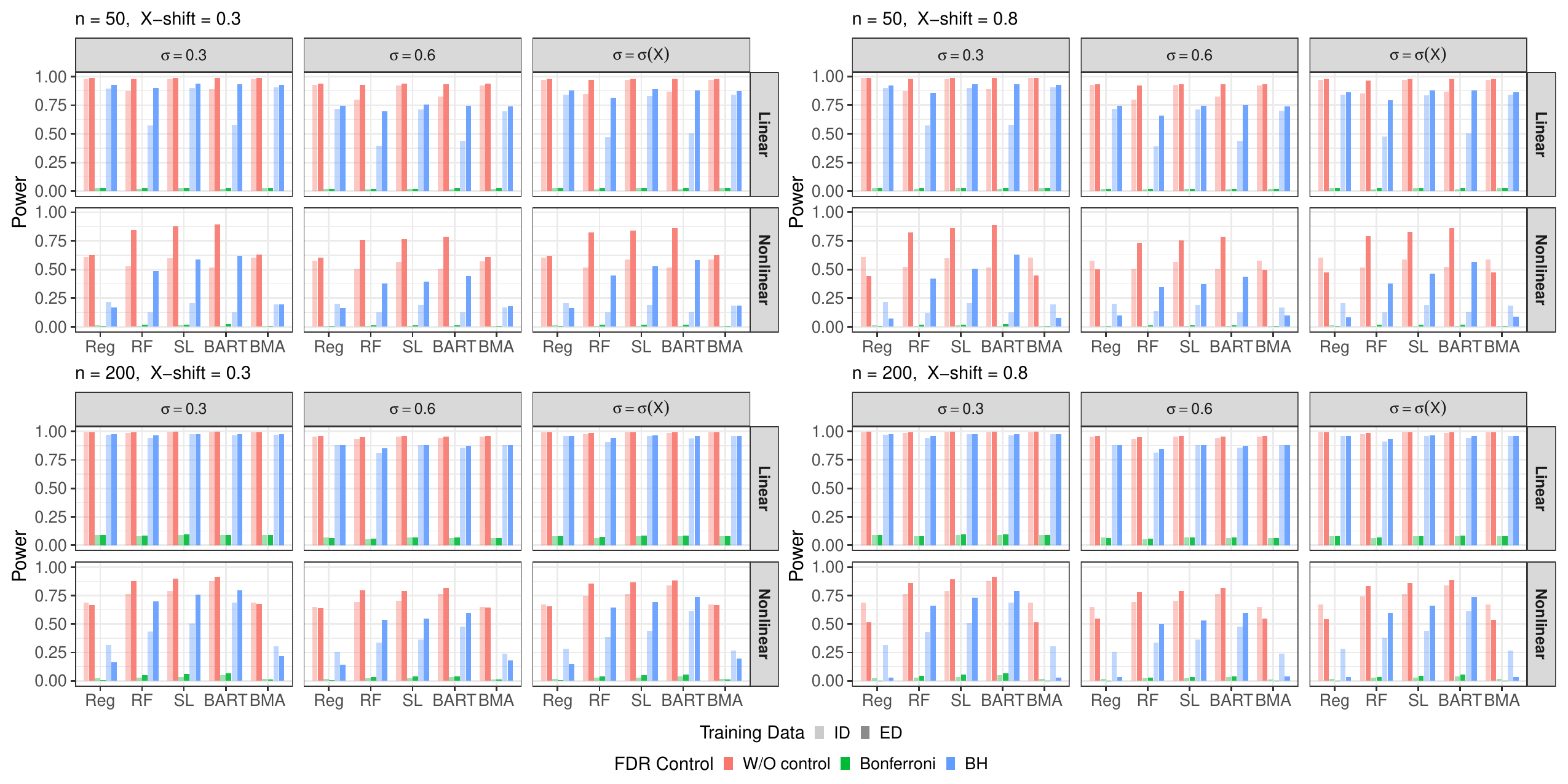}
    \caption{Power across different scenarios given outcome drift ($b=0.4$)}
    \label{fig:Power_all_Yshift_0.4}
\end{figure}

\section{Additional case study results}\label{app:case_results}
In this section, we additionally present the case study results using BART, which has demonstrated strong performance and robustness in our simulation studies. Figure~\ref{fig:heatmap_bart_0.35} shows that when only RCT data are used for model training, the estimated pattern suggests that patients around age 80 are more likely to benefit from limited resection, with a relatively abrupt change in the surface. In contrast, when external RWD is incorporated into model training, the estimated pattern becomes noticeably smoother, revealing a more gradual gradient with respect to age rather than a discrete jump observed under the ID approach. This smoother trend suggests that the benefit of limited resection increases progressively with age.

Across both ID and ED approaches, there is consistent evidence that patients with smaller tumor sizes are more likely to benefit from limited resection. Importantly, incorporating external data through ED helps stabilize the estimation and recover a more interpretable and clinically plausible pattern. The resulting conclusions are more aligned with prior findings in the literature \citep{Altorki2023,Mao2025pub}. Figure~\ref{fig:sep_bart} shows a pattern similar to that observed with the Super Learner, further illustrating that incorporating external data can enhance the ability to distinguish signal from noise near the decision boundary. Therefore, there is an improved separation between treatment-benefiting subgroups.

\begin{figure}[ht]
    \centering
    \includegraphics[width=\linewidth]{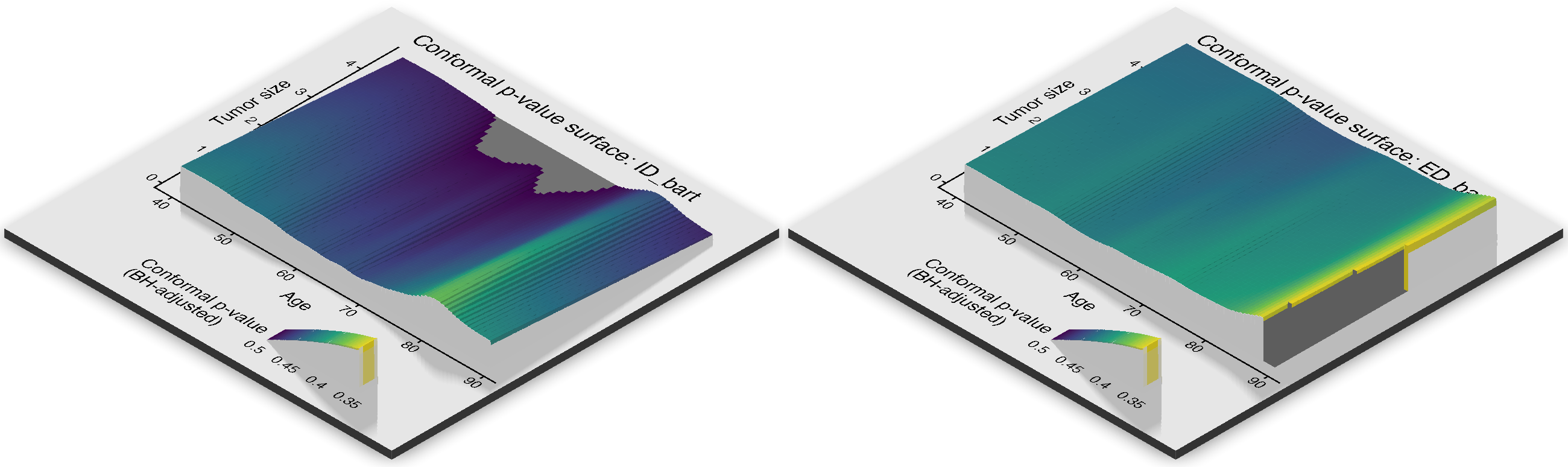}
    \caption{Comparison of conformal $p$-value surfaces under ID and ED when limited resection is treatment of interest and BART is base learner}
    \label{fig:heatmap_bart_0.35}
\end{figure}

\begin{figure}[ht]
    \centering
    \includegraphics[width=\linewidth]{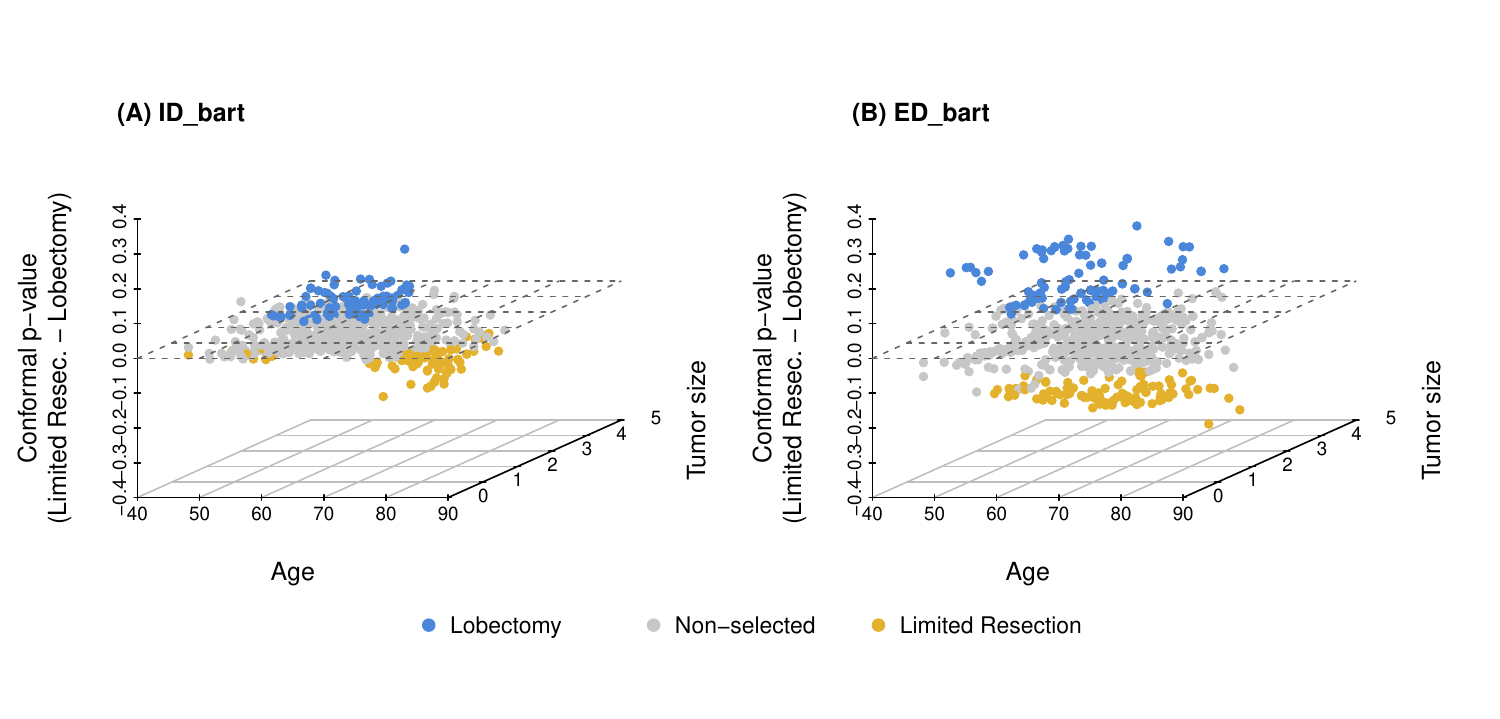}
    \caption{Distribution of conformal $p$-values regarding age and tumor sizes when using BART as the base learner}
    \label{fig:sep_bart}
\end{figure}

\section{Empirical assessment of stochastic ordering}\label{app:empirical_stoch}
A key condition for using pseudo-outcome-based conformal inference for individual treatment effects is that the conformal scores computed from the pseudo-outcomes should be conservative relative to the oracle conformal scores that would be obtained if the true individual treatment effects were observed. Following the stochastic-ordering framework of \citep{alaa2024conformal}, this means that the pseudo-outcome conformal scores should be stochastically larger, or more variable in an appropriate sense, than the oracle scores. Intuitively, if the pseudo-score distribution has a larger upper quantile than the oracle-score distribution, then the conformal threshold calibrated using pseudo-outcomes will be at least as conservative as the threshold that would have been calibrated using the unobserved oracle ITEs.

\subsection{Assessment under the main simulation settings}\label{app:empirical_stoch_main}
Here, we provide an empirical assessment of this stochastic-ordering condition across our simulation settings, where the outcome is continuous. This assessment is not used as part of the proposed procedure; rather, it serves as a diagnostic check to understand whether the pseudo-outcome transformation behaves as expected in the settings considered in our simulations.

\subsubsection{Clip conformal score}\label{app:empirical_stoch_main_clip}
For each simulation scenario, we compare the empirical CDFs of two types of calibration scores. The first is the pseudo-outcome-based score, $V_i' = M\mathbb{I}\{Y_i'>0\} - \hat{\tau}(X_i),$ where $Y_i'$ denotes the DR-learner pseudo-ITE, $\hat{\tau}(X_i)$ is the fitted ITE model, and $M$ is a sufficiently large constant satisfying $M > 2\sup_x |\hat\tau(x)|$. The second is the oracle score, $V_i^* = M\mathbb{I}\{(Y_i(1)-Y_i(0))>0\}-\hat{\tau}(X_i),$ which is only available in simulation because the true potential outcomes are known. For each iteration, we compute the empirical CDFs of $V_i'$ and $V_i^*$ using the same calibration subjects. We then average these empirical CDFs over 200 simulation iterations to obtain a stable visualization of the score distributions.

Figures \ref{fig:ECDF_Stochastic_Ordering_clip_n50_etaj0.3} - \ref{fig:ECDF_Stochastic_Ordering_clip_n200_etaj0.8} present the resulting empirical CDFs under different sample sizes and levels of covariate shift. Overall, the oracle-score CDF is generally above, or shifted to the left of, the pseudo-score CDF. This means that the oracle scores accumulate more quickly at smaller score values, while the pseudo-outcome scores tend to have larger upper quantiles. This is consistent with the conservative direction for conformal calibration that \cite{alaa2024conformal} validated. The threshold computed from pseudo-ITE scores is expected to be at least as large as the infeasible oracle threshold, which preserves validity after replacing the unobserved ITE with a pseudo-ITE.

The degree of separation between the two CDF curves also provides useful information about efficiency. When the pseudo-score CDF is close to the oracle-score CDF, the pseudo-outcome transformation approximates the oracle ITE well. When the pseudo-score CDF is much lower than the oracle-score CDF, the pseudo-outcomes induce more variability in the conformal scores and might lead to lower power in beneficiary identification. This finding is consistent with \cite{alaa2024conformal}, where the empirical gap between pseudo-score and oracle-score CDFs can explain differences in coverage and efficiency across conformal meta-learners. In addition, the empirical assessment is also aligned with the main simulation results. Base learners with stronger prediction performance, i.e., higher power with FDR controlled, such as Super Learner, tend to have pseudo-score distributions closely follow the oracle-score distribution. In contrast, when the pseudo-score distribution is more separated from the oracle distribution, the procedure remains conservative but selects fewer true beneficiaries.

\begin{figure}
    \centering
    \includegraphics[width=\linewidth]{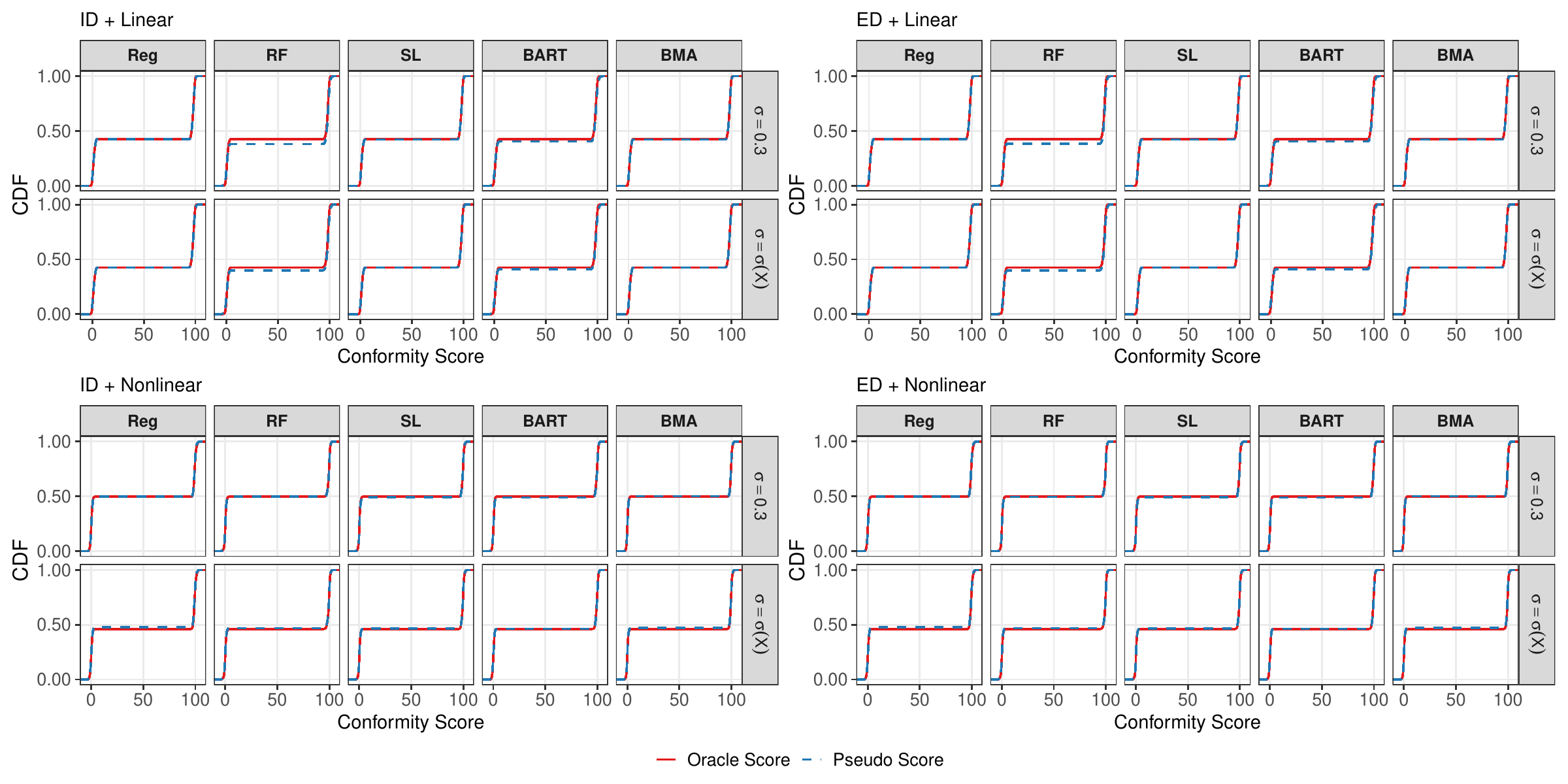}
    \caption{Empirical CDFs of pseudo-outcome and oracle clip conformal scores under the main simulation setting with $n=50$ and covariate shift $\eta=0.3$}
    \label{fig:ECDF_Stochastic_Ordering_clip_n50_etaj0.3}
\end{figure}
\begin{figure}
    \centering
    \includegraphics[width=\linewidth]{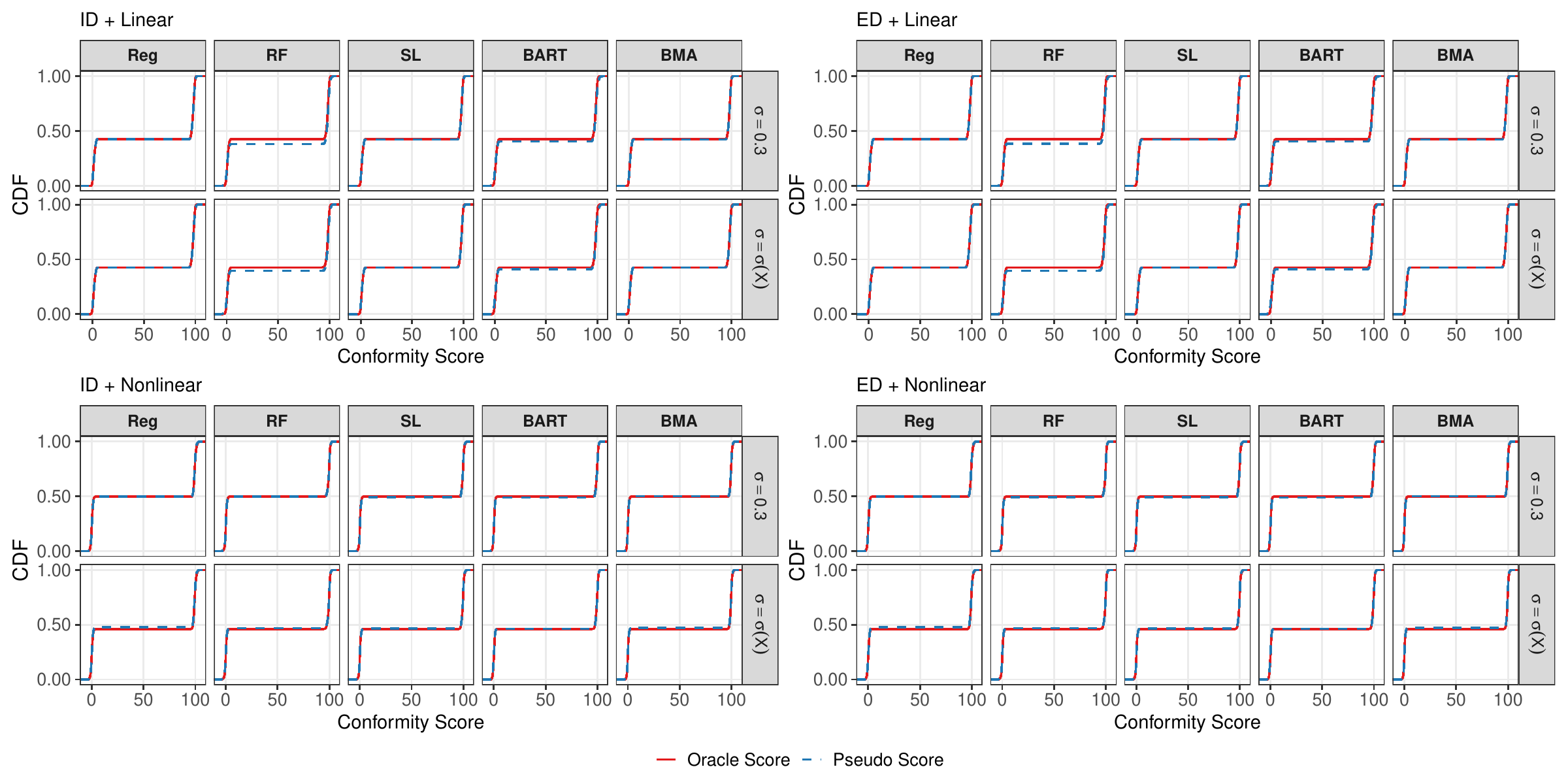}
    \caption{Empirical CDFs of pseudo-outcome and oracle clip conformal scores under the main simulation setting with $n=50$ and covariate shift $\eta=0.8$}
    \label{fig:ECDF_Stochastic_Ordering_clip_n50_etaj0.8}
\end{figure}
\begin{figure}
    \centering
    \includegraphics[width=\linewidth]{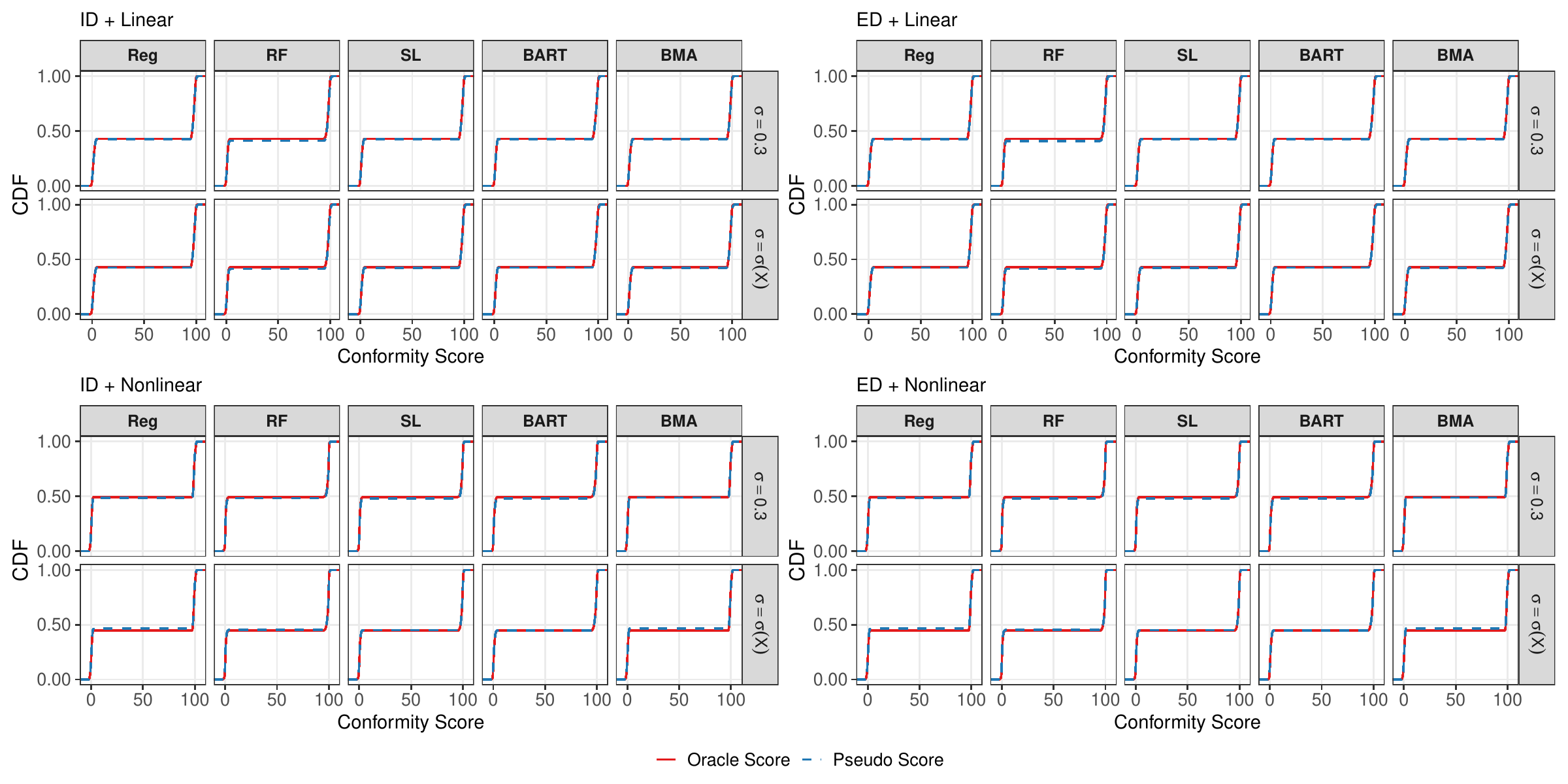}
    \caption{Empirical CDFs of pseudo-outcome and oracle clip conformal scores under the main simulation setting with $n=200$ and covariate shift $\eta=0.3$}
    \label{fig:ECDF_Stochastic_Ordering_clip_n200_etaj0.3}
\end{figure}
\begin{figure}
    \centering
    \includegraphics[width=\linewidth]{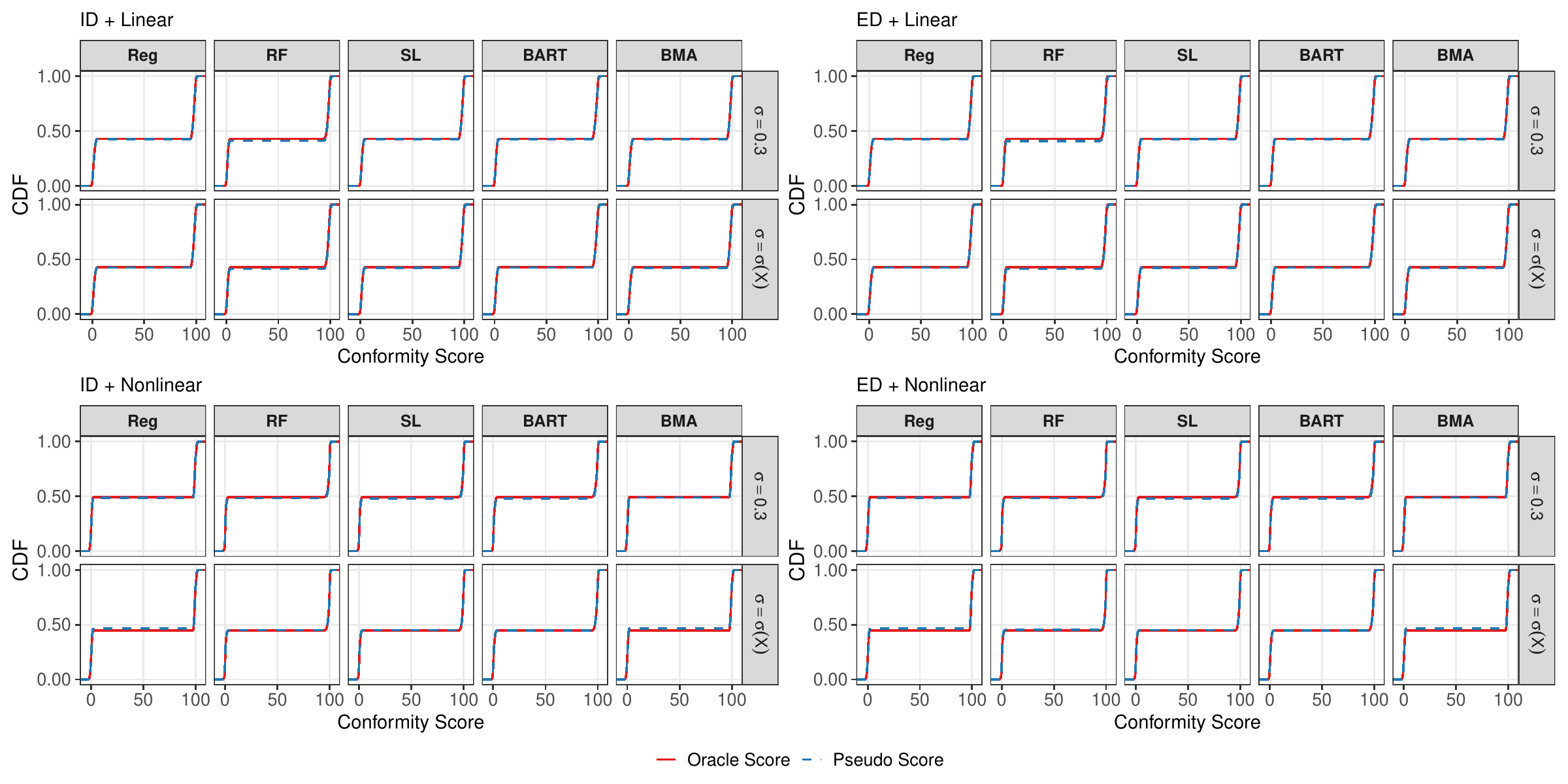}
    \caption{Empirical CDFs of pseudo-outcome and oracle clip conformal scores under the main simulation setting with $n=200$ and covariate shift $\eta=0.8$}
    \label{fig:ECDF_Stochastic_Ordering_clip_n200_etaj0.8}
\end{figure}

\subsubsection{Residual conformal score}\label{app:empirical_stoch_main_res}
We further examine the empirical CDFs when using the residual conformal score. The pseudo-outcome-based score is defined as $V_i' = Y_i' - \hat{\tau}(X_i)$, and the oracle score is $ V_i^* = Y_i(1)-Y_i(0)-\hat{\tau}(X_i)$. 

Figures \ref{fig:ECDF_Stochastic_Ordering_res_n50_etaj0.3} - \ref{fig:ECDF_Stochastic_Ordering_res_n200_etaj0.8} illustrate the empirical distributions of the residual conformal scores across different simulation settings. Overall, the pseudo-score and oracle-score CDFs are very close. In some settings, the two curves may cross slightly when using Random Forest and BART as base learners. This pattern is consistent with the stochastic-ordering results in \citep{alaa2024conformal}. For residual-type signed-distance scores, the DR-learner is associated with a weaker second-order stochastic ordering, rather than a strict first-order dominance pattern. Therefore, together with the clip-score results, this supports the use of both score choices in our framework when the ITE is the prediction target.

\begin{figure}
    \centering
    \includegraphics[width=\linewidth]{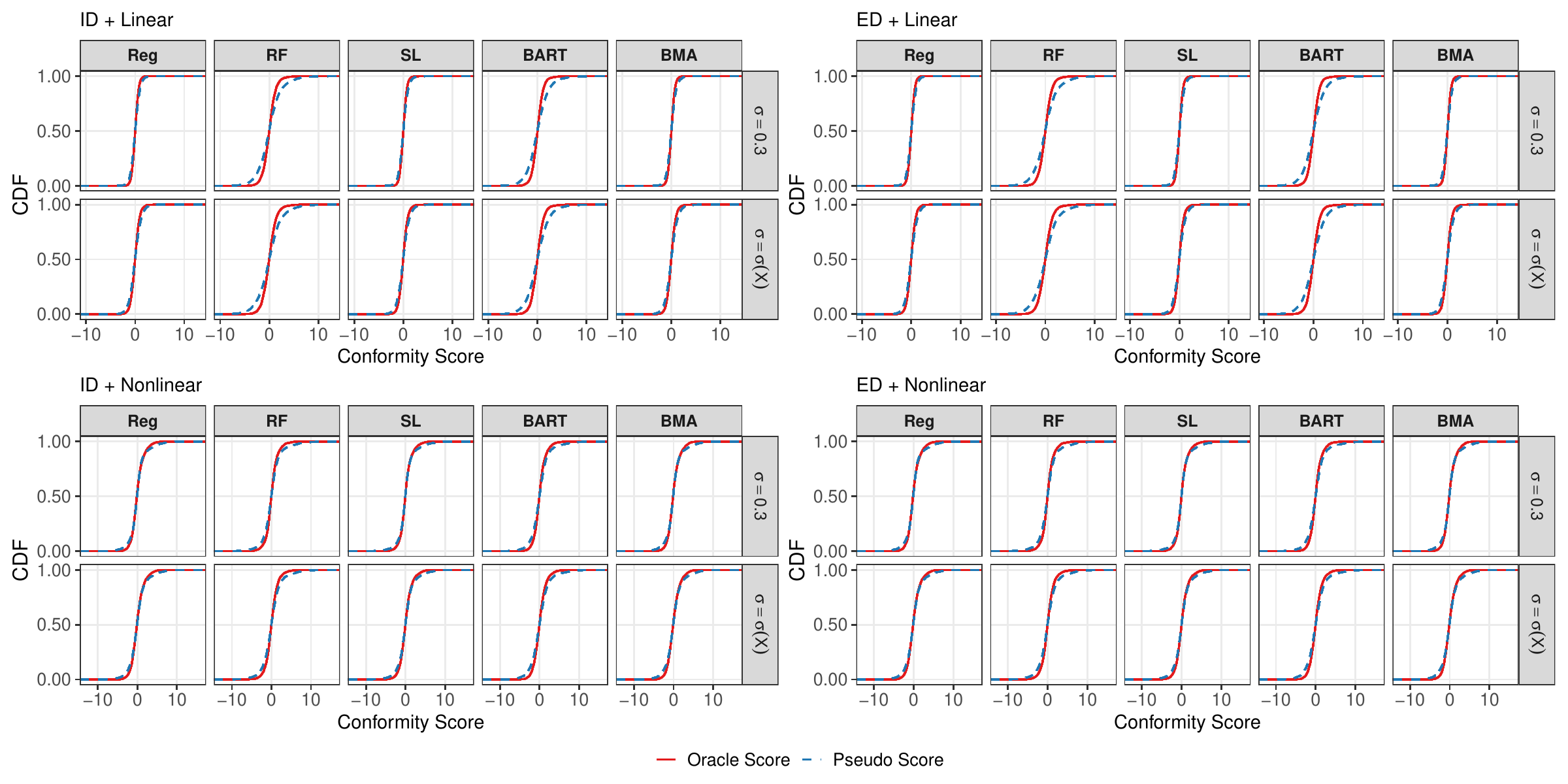}
    \caption{Empirical CDFs of pseudo-outcome and oracle residual conformal scores under the main simulation setting with $n=50$ and covariate shift $\eta=0.3$}
    \label{fig:ECDF_Stochastic_Ordering_res_n50_etaj0.3}
\end{figure}
\begin{figure}
    \centering
    \includegraphics[width=\linewidth]{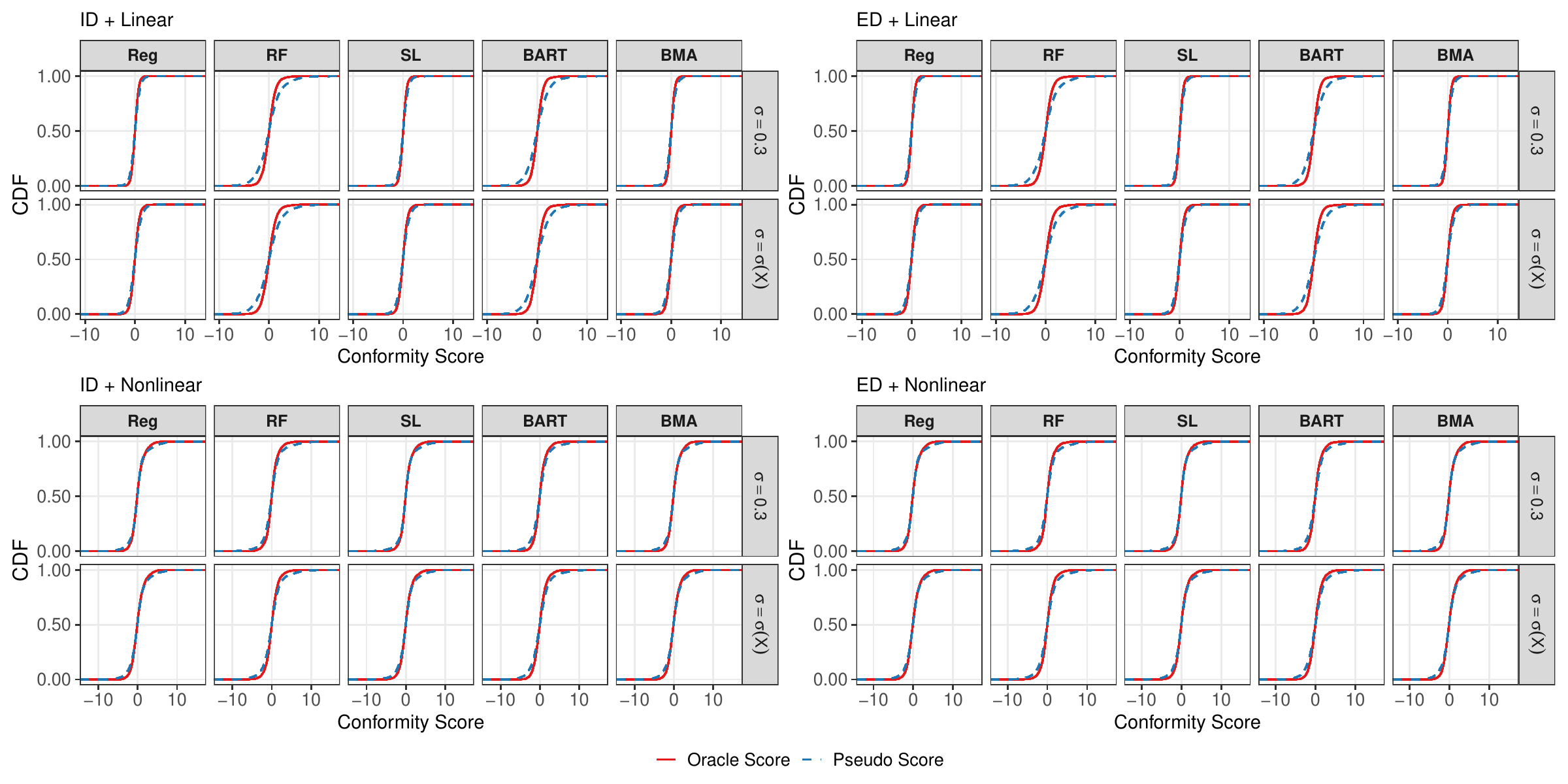}
    \caption{Empirical CDFs of pseudo-outcome and oracle residual conformal scores under the main simulation setting with $n=50$ and covariate shift $\eta=0.8$}
    \label{fig:ECDF_Stochastic_Ordering_res_n50_etaj0.8}
\end{figure}
\begin{figure}
    \centering
    \includegraphics[width=\linewidth]{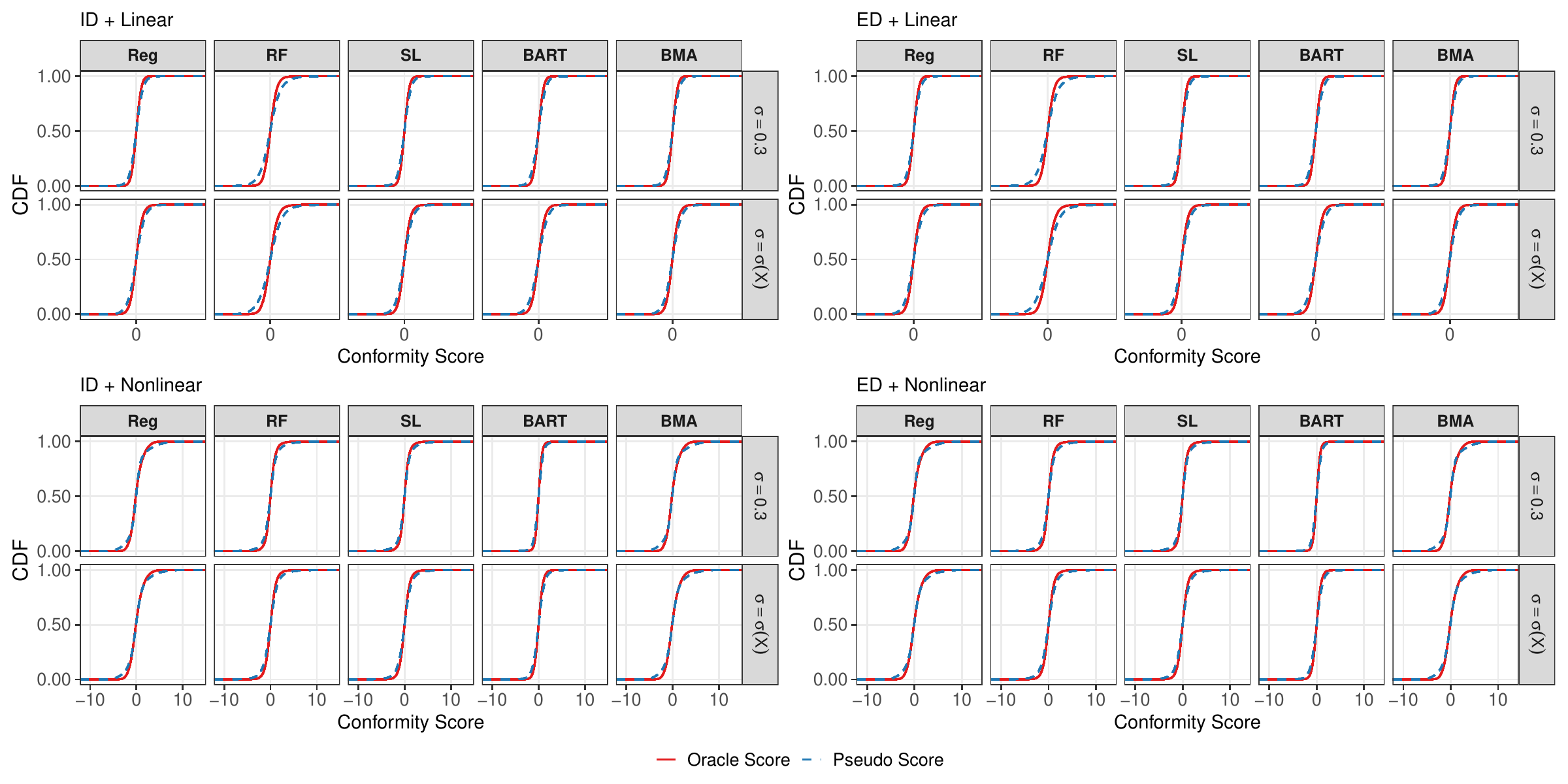}
    \caption{Empirical CDFs of pseudo-outcome and oracle residual conformal scores under the main simulation setting with $n=200$ and covariate shift $\eta=0.3$}
    \label{fig:ECDF_Stochastic_Ordering_res_n200_etaj0.3}
\end{figure}
\begin{figure}
    \centering
    \includegraphics[width=\linewidth]{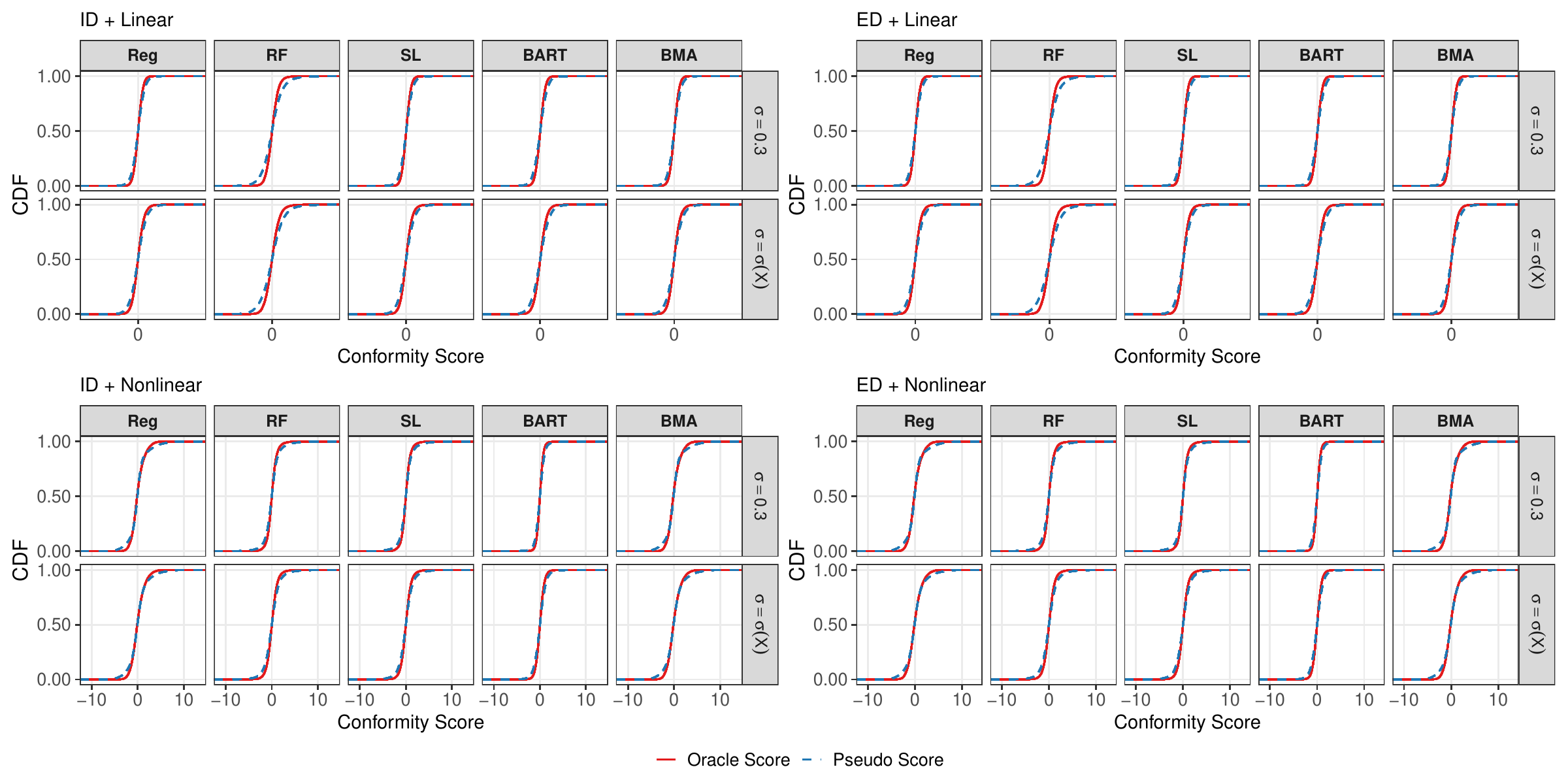}
    \caption{Empirical CDFs of pseudo-outcome and oracle residual conformal scores under the main simulation setting with $n=200$ and covariate shift $\eta=0.8$}
    \label{fig:ECDF_Stochastic_Ordering_res_n200_etaj0.8}
\end{figure}

\subsection{Assessment under case-study-like simulation settings}\label{app:empirical_stoch_caseStudy}
In this section, we assess the empirical CDFs of conformal scores under a case-study-like simulation setting. Since the true ITEs are not observed in the real case study, the oracle conformal scores cannot be directly computed from the real data. To make this comparison possible, we generated simulated datasets that mimic the main features of the case study, including the covariate distribution, outcome pattern, and censoring rate. This allows us to evaluate whether the pseudo-outcome transformation behaves similarly to the oracle ITE in a setting closer to the real-data application. In the case study, the outcome is subject to censoring. We therefore first transform the survival outcome into pseudo-RMST, so that the transformed outcome can be analyzed as a continuous outcome. We then follow the same empirical CDF assessment used in the previous section. Specifically, we compare the pseudo-outcome-based conformal scores with the oracle conformal scores computed from the known simulated potential outcomes. 

Figure \ref{fig:ECDF_Stochastic_Ordering_clip_case_study} shows the results when using the clip conformal score. The pseudo-score CDFs almost overlap with the oracle-score CDFs, suggesting that the pseudo-RMST transformation provides a close approximation to the oracle ITE score distribution in this case-study-like setting. When using the residual conformal score in Figure \ref{fig:ECDF_Stochastic_Ordering_res_case_study}, the CDF curves show second-order stochastic ordering pattern. This is consistent with the findings in \cite{alaa2024conformal}, where the CDF curves may cross oracle CDF curve is higher than pseudo CDF curve among some conformal scores. Overall, these case-study-like simulations suggest that the pseudo-RMST transformation does not substantially distort the conformal score distribution.

\begin{figure}
    \centering
    \includegraphics[width=\linewidth]{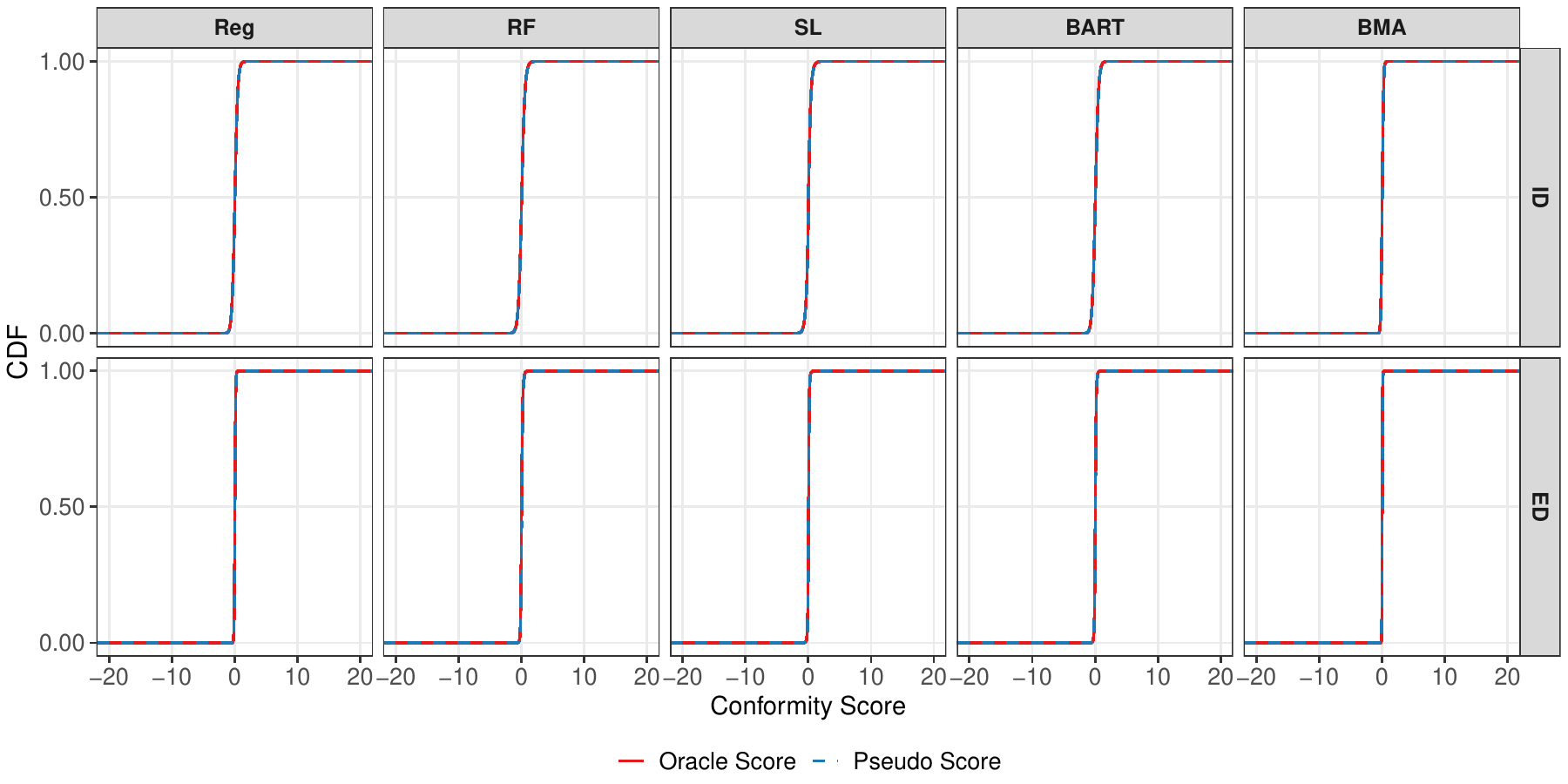}
    \caption{Empirical CDFs of pseudo-outcome and oracle clip conformal scores under the case study setting}
    \label{fig:ECDF_Stochastic_Ordering_clip_case_study}
\end{figure}

\begin{figure}
    \centering
    \includegraphics[width=\linewidth]{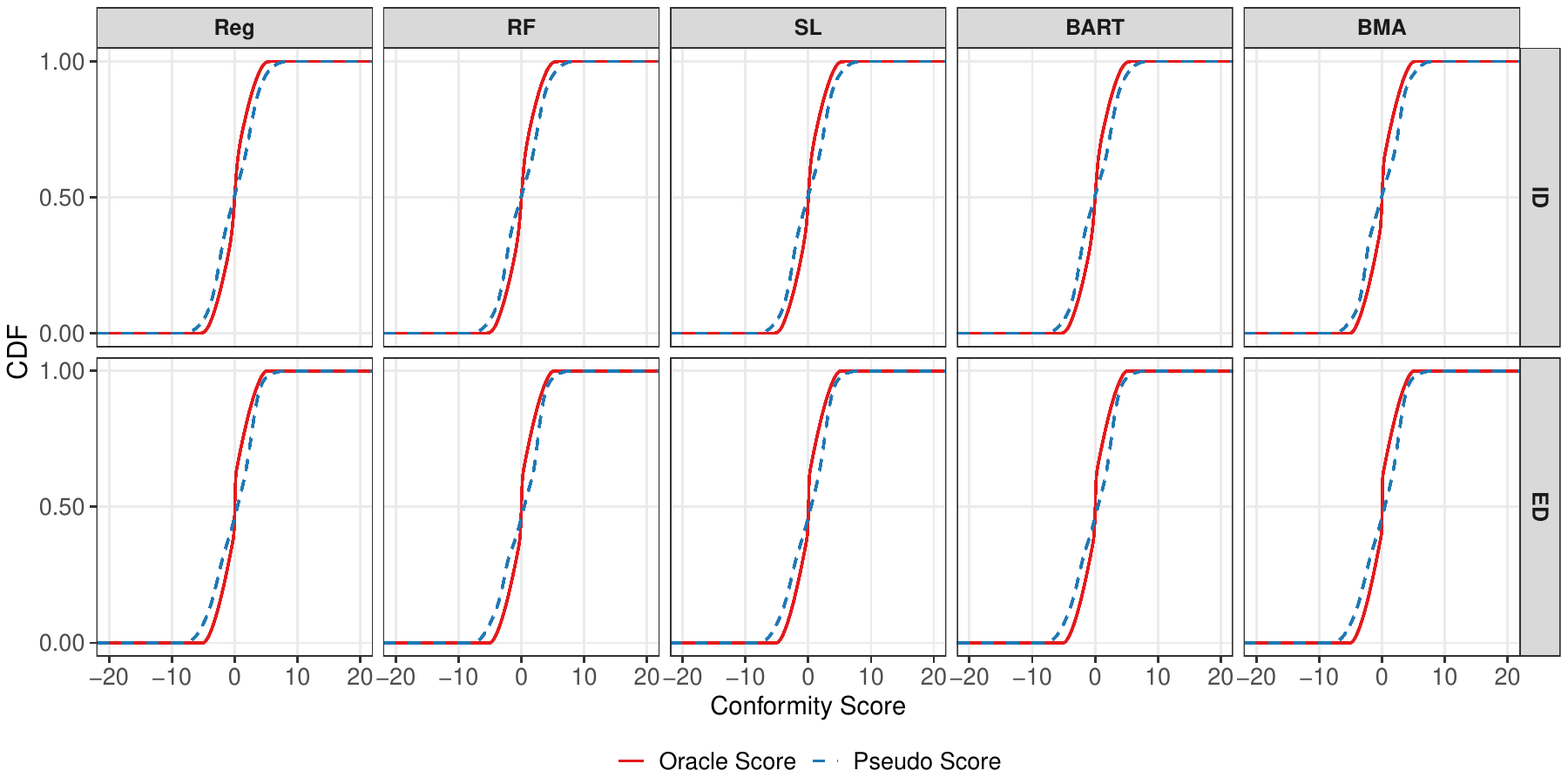}
    \caption{Empirical CDFs of pseudo-outcome and oracle residual conformal scores under the case study setting}
    \label{fig:ECDF_Stochastic_Ordering_res_case_study}
\end{figure}

\section{Pseudo-code for weighted conformal}\label{app:weighted_conformal}
In Algorithm \ref{alg:weighted_conformal_screening}, we provide pseudo-code for incorporating weighted conformal inference into our framework, with the goal of extending beneficiary identification to a new candidate population under covariate shift. The key idea is to reweight the calibration scores toward the target candidate population \citep{tibshirani2019conformal}. 
\begin{algorithm}[ht]
\small
\caption{Weighted conformal screening under covariate shift}
\label{alg:weighted_conformal_screening}
\KwData{
$\mathcal D_{\rm RCT}$; optional $\mathcal D_{\rm RWD}$; target covariates
$\mathcal D_{\rm Test}=\{X_{n+m+j}\}_{j=1}^{n_t}$; thresholds $\{c_j\}_{j=1}^{n_t}$;
FDR level $q$; monotone score $V(x,y)$.
}
\KwResult{Selected beneficiary set $\mathcal R$.}

Split $\mathcal D_{\rm RCT}$ into $\mathcal R_{\rm Train}$ and $\mathcal D_{\rm Calib}$.\;

Let $\mathcal D_{\rm Train}=\mathcal R_{\rm Train}\cup \mathcal D_{\rm RWD}$. Construct pseudo-outcomes
$Y_i'$ using a DR-learner, and fit $\hat\tau(x)$ by regressing $Y'$ on $X$ in $\mathcal D_{\rm Train}$.\;

Pool covariates from $\mathcal D_{\rm Calib}$ and $\mathcal D_{\rm Test}$. Define $S=0$ for calibration subjects and $S=1$ for target candidates. Fit a sampling model $S\sim X$ and set
\[
\hat p(x)=\Pr(S=1\mid X=x),\qquad
\hat w(x)=\frac{\hat p(x)}{1-\hat p(x)}.
\]
Here, $\hat w(x)$ is the individual weight to address covariate shift between the target candidate population and the RCT calibration population\;

Compute calibration scores and candidate null-imputed scores:
\[
V_i=V(X_i,Y_i'),\ i\in\mathcal D_{\rm Calib},
\qquad
\hat V_j=V(X_{n+m+j},c_j),\ j=1,\ldots,n_t.
\]
The score $\hat V_j$ is computed under the null boundary $H_{0j}:\tau(X_{n+m+j})\le c_j$\;

For each candidate $j$, define the target-population-weighted conformal probabilities
\[
\pi_i^{(j)}
=
\frac{\hat w(X_i)}
{\sum_{\ell\in\mathcal D_{\rm Calib}}\hat w(X_\ell)+\hat w(X_{n+m+j})},
\quad
\pi_j^{\rm test}
=
\frac{\hat w(X_{n+m+j})}
{\sum_{\ell\in\mathcal D_{\rm Calib}}\hat w(X_\ell)+\hat w(X_{n+m+j})}.
\]

Compute the weighted conformal $p$-value
\[
p_j^w
=
\sum_{i\in\mathcal D_{\rm Calib}}\pi_i^{(j)}\mathbb{I}(V_i<\hat V_j)
+
U_j\left\{\pi_j^{\rm test}
+
\sum_{i\in\mathcal D_{\rm Calib}}\pi_i^{(j)}\mathbb{I}(V_i=\hat V_j)
\right\},
\quad U_j\sim{\rm Unif}(0,1).
\]
This provides a weighted distribution that better represents the target candidate population\;

Apply BH to $\{p_j^w\}_{j=1}^{n_t}$:
\[
k^\ast=\max\left\{k:p_{(k)}^w\le \frac{qk}{n_t}\right\},
\quad
\mathcal R=\left\{j:p_j^w\le \frac{qk^\ast}{n_t}\right\}.
\]
\end{algorithm}

\end{document}